\documentclass[12pt]{article}
\usepackage[latin9]{inputenc}
\usepackage{geometry}
\usepackage{array}
\usepackage{booktabs}
\usepackage{calc}
\usepackage{amsmath}
\usepackage{amsthm}
\usepackage{amssymb}
\usepackage{stmaryrd}
\usepackage{setspace}
\usepackage[authoryear]{natbib}
\usepackage[unicode=true,
 bookmarks=false,
 breaklinks=false,pdfborder={0 0 1},colorlinks=false]
 {hyperref}
\hypersetup{
 colorlinks,citecolor=blue,pdftex}

\usepackage{booktabs}
\usepackage{siunitx}
\usepackage{threeparttable}
\usepackage{adjustbox}

\makeatletter

\usepackage{amsfonts}
\usepackage{amsthm}
\usepackage{lscape}
\usepackage{appendix}
\usepackage{dcolumn}
\usepackage{rotating}
\usepackage{comment}
\usepackage{color}

\usepackage{diagbox}

\usepackage{changepage}
\usepackage{booktabs}

\usepackage{courier}

\newcolumntype{d}[1]{D{.}{.}{#1}}
\newcolumntype{t}[1]{D{,}{,}{#1}}
\newcolumntype{i}[1]{D{.}{}{#1}}
\newtheorem{theorem}{Theorem}[section]

\newtheorem{algorithm}{Algorithm}[section]
\newtheorem{assumption}{Assumption}[section]

\newtheorem{corollary}{Corollary}[section]

\newtheorem{definition}{Definition}[section]
\newtheorem{example}{Example}
\newtheorem{lemma}{Lemma}[section]

\newtheorem{proposition}{Proposition}[section]
\newtheorem{remark}{Remark}[section]

\theoremstyle{plain}

\newtheorem{innercustomprompt}{Prompt}\newenvironment{myprompt}[1]{\renewcommand\theinnercustomprompt{#1}\innercustomprompt}{\endinnercustomprompt}

\numberwithin{equation}{section}

\usepackage{xcolor}

\begin{document}
\title{How Well Do LLMs Predict Human Behavior? \\ A Measure of their Pretrained Knowledge\footnote{For helpful discussions, we thank Tom Cunningham, Matt Gentzkow, Guido Imbens, and Pat Kline.}}
\author{Wayne Gao\footnote{Department of Economics, University of Pennsylvania} \quad \quad Sukjin Han\footnote{School of Economics, University of Bristol} \quad \quad Annie Liang\footnote{Department of Economics, Northwestern University} }
\maketitle

\begin{abstract}

Large language models (LLMs) are increasingly used to predict human behavior. We propose a measure for evaluating how much knowledge a pretrained LLM brings to such a prediction: its \emph{equivalent sample size}, defined as the amount of task-specific data needed to match the predictive accuracy of the LLM. We estimate this measure by comparing the prediction error of a fixed LLM in a given domain to that of flexible machine learning models trained on increasing samples of domain-specific data. We further provide a statistical inference procedure by developing a new asymptotic theory for cross-validated prediction error. Finally, we apply this method to the Panel Study of Income Dynamics. We find that LLMs encode considerable predictive information for some economic variables but much less for others, suggesting that their value as substitutes for domain-specific data differs markedly across settings.
\end{abstract}

\section{Introduction}

Large language models (LLMs) are increasingly used in  economics  as predictive tools---both to generate synthetic responses in place of human subjects \citep{horton2023large,AnthisEtAl2025LLMSocialSimulations}, and to forecast economic outcomes directly \citep{HewittEtAl2024PredictingSocialScience,fariaeCastroLeibovici2024AIInflation,chan-lauEtAl2025ForecastingFedFundsAI}. Their appeal in these roles is obvious: A pretrained LLM embeds a vast amount of information and can be deployed at negligible cost, often in settings where collecting new, domain-specific human data would be expensive or infeasible.

What remains unclear is how to assess the \emph{quality} of these predictions. This paper proposes a measure that quantifies the domain-specific value of LLMs in an interpretable unit: the amount of human data they substitute for. Specifically, we ask how much human data would be required for a conventional model trained on that data to match the predictive performance of the pretrained LLM in that domain. We define the LLM's \emph{equivalent sample size} to be the smallest size of training data required for the model to match the LLM. 

For some prediction tasks, LLMs may perform comparably to models trained on large datasets, making them a useful surrogate for data collection. For others, LLMs may be so poorly suited that even models trained on a small amount of domain-specific data quickly outperform them. Quantifying the LLM's equivalent sample size can thus inform the need for data collection and the interpretation of LLM-based empirical results across domains.

Beyond introducing this measure, the paper develops a method for estimating the equivalent sample size and for conducting statistical inference. In this part of the paper, we construct a confidence interval for our proposed measure via sequential hypothesis testing, whose validity is shown by developing a novel asymptotic theory for cross-validated estimators. Finally, we illustrate the measure by evaluating LLM predictions of variables in the Panel Study of Income Dynamics, where we find substantial heterogeneity in the extent to which LLMs can proxy for observed data across prediction tasks.

The paper proceeds as follows. Section \ref{sec:Framework} formalizes the setting and introduces our proposed measure. We consider general prediction problems in which an analyst seeks to predict an outcome $Y$ given covariates $X$. A \emph{prediction rule} is a mapping from covariates to outcomes, and the error of any such rule is evaluated by its expected loss under a specified loss function. The standard approach for learning a prediction rule is to collect a dataset of $(X,Y)$ observations and train a model or machine learning algorithm on that data.

Large language models (LLMs) offer an alternative. Rather than training on domain-specific data, the analyst can describe the prediction problem to a pretrained LLM and request predictions of $Y$ given $X$. While the LLM is itself trained on data, we focus on settings in which the analyst deploys a fixed, pretrained model without further domain-specific fine-tuning, as is common in the literature (see Section \ref{sec:RelatedLit}). We refer to the resulting mapping from covariates to predicted outcomes as the \emph{LLM prediction rule}.

We then compare the error of the LLM prediction rule to that of a fixed comparator trained on datasets of increasing size. This comparator may be an off-the-shelf machine-learning method or a structural economic model; for brevity, we refer to it simply as an algorithm. Training the algorithm on datasets of different sizes yields a sequence of prediction rules that trace out an error curve. We define the \emph{equivalent sample size} to be the smallest training sample size at which the algorithm's error weakly improves upon that of the LLM prediction rule. This quantity can be interpreted as the amount of data required to train a model that is as informative as the pretrained LLM for the prediction task at hand.

Section \ref{sec:Procedure} develops an estimator and inference procedure for the equivalent sample size. We estimate the LLM's prediction error using its performance on the full sample, and estimate the error curve of the comparator algorithm via block-out cross-validation. We finally introduce a sequential testing procedure that compares the LLM and algorithm across an ordered sequence of training sizes and yields a one-sided confidence interval for the equivalent sample size.

Section \ref{sec:InferenceTheory} provides theoretical foundations for this inference approach and establishes its validity under the mild assumption that the algorithm's error is weakly reducing in the size of its training data. Because the sequence of hypotheses comparing the algorithm and the LLM is nested, the sequential procedure controls coverage without requiring multiple-testing corrections, thus avoiding the conservativeness that typically arises in multiple hypothesis testing. This feature is especially important in our setting, where we want to compare algorithm performance across a large range of training sample sizes.  We finally derive a central limit theorem for the block-out cross-validated risk estimator under a fixed-training-size asymptotic regime, and construct a consistent estimator of its asymptotic variance. Together, these results establish that the proposed procedure delivers a valid one-sided confidence interval for the equivalent sample size with correct asymptotic coverage.

Section \ref{sec:Application} applies our proposed measure to evaluate LLM prediction of variables in the Panel Study of Income Dynamics (PSID). The PSID is one of the longest-running longitudinal household surveys in the U.S., and is frequently used in empirical research. Using the 2021 wave of the Panel Study of Income Dynamics, we study a suite of prediction tasks in which respondents' demographic, labor-market, and household covariates are used to forecast hourly wages, homeownership, drinking, and smoking. We construct LLM prediction rules by translating each covariate vector into a natural-language persona description and querying the model. To limit concerns about data leakage, we emphasize post-cutoff evaluation by using static open-source models with documented training cutoffs. The estimated equivalent sample size varies substantially across outcomes---e.g., approximately 20 observations for predicting hourly wages to 600 observations for predicting homeownership---highlighting substantial heterogeneity in the extent to which pretrained knowledge substitutes for domain-specific data.

Finally, Section \ref{sec:CATE} extends our evaluation framework to assessing LLM performance when prompted to infer counterfactual outcomes or treatment effects. Under an assumption of unconfoundedness, we define equivalent sample sizes for both the average potential outcome conditional on covariates and the conditional average treatment effect (CATE). Although the true CATE is unknown, a prediction framework with a transformed outcome enables us to calculate the risk for CATE for both the LLM and ML and thus the equivalent sample size.

\subsection{Related Literature} \label{sec:RelatedLit}

\paragraph{LLMs as Human Surrogates.}

LLMs are increasingly used as \emph{surrogate} experimental and survey subjects in place of human participants. Researchers prompt LLMs with experimental or survey scenarios and treat the resulting responses as proxies for human behavior. This approach has been applied in survey settings \citep{argyle2023out,jansen2023employing,sun2024random} and in experimental settings \citep{aher2023using,hewitt2024predicting,park2024generative,BohrenHullImas2025}. Related work explores broader applications of LLMs as simulated economic or social agents, including their use in market simulations and policy analysis \citep{horton2023large,manning2024automated}.

Despite this growing adoption, there is little consensus on how to assess the quality or value of LLMs as human surrogates. Much of the existing literature evaluates performance in task-specific or descriptive terms---asking whether LLM outputs resemble human responses in a given domain---without a unifying metric that quantifies what is gained or lost by substituting LLMs for human data. At the same time, several papers caution against uncritical reliance on LLMs, documenting systematic biases, domain failures, and sensitivity to prompting choices \citep{gao2025take,ludwig2025large}. This paper contributes a general measure for evaluating the appropriateness of LLM predictions: the amount of human data the LLM effectively substitutes for.

\paragraph{Model Evaluation.}

Our paper is closely related to \citet{HuangWuWang2025LLMSurveyWorth}, who also interpret the value of LLM outputs in terms of an equivalent number of human observations. In their setting, the effective sample size arises from an uncertainty-quantification problem: how many synthetic LLM responses are required to construct confidence intervals with valid coverage for human population parameters. In contrast, we treat the LLM as a fixed prediction rule and define the equivalent sample size as the amount of real, domain-specific data required for a supervised learning algorithm to match the LLM's predictive performance.

Additionally, our approach is in the spirit of \citet{ErevRothSlonimBarron2007}'s notion of the \emph{equivalent number of observations (ENO)}, which evaluates game-theoretic models by asking how many subject observations would be required for the sample average of observed behavior to predict as well as the theory. The objects of comparison are however different: In their setting, the object being evaluated---a behavioral or equilibrium model---is itself estimated from data, whereas we treat the LLM as a fixed prediction rule that is not trained or tuned on the domain sample. Additionally, their comparator is the sample average of observed behavior, while ours is the prediction of a supervised learning algorithm. Thus our measure applies to a substantially broader set of prediction tasks.

Finally, our work is related more broadly to work on evaluating economic models \citep{FudenbergKleinbergLiangMullainathan2022,FudenbergGaoLiangForthcoming,AndrewsFudenbergLeiLiangWu2025}, where here the ``model'' is the pretrained LLM. In particular, our benchmarking of LLM performance against that of a comparator algorithm is similar to \citet{FudenbergKleinbergLiangMullainathan2022}, although our focus on the equivalent \emph{sample size} of the comparator algorithm is entirely new. 

\paragraph{Asymptotic Theory with Cross-Validation.}

A growing literature develops asymptotic theory for cross-validated prediction error to address the dependence induced by overlapping training samples across folds \citep{austern2020asymptotics,bates2024cross,fava2025training,lei2025modern}. Our results in Section 3.4 for block-out cross validation contribute to this literature, and are in particular related to \cite{lei2025modern} and \cite{fava2025training}. The key difference in our analysis is that we consider an asymptotic regime  where the training sample size is held fixed. \cite{lei2025modern} and \cite{fava2025training} instead work with a standard asymptotic regime where the training sample size grows large, which is natural for typical machine learning applications. Inference for our proposed measure demands a different  asymptotic regime, since we implement hypothesis testing with respect to a given number of training observations.

\section{Framework} \label{sec:Framework}

Section \ref{sec:Setup} introduces our framework: we study prediction problems in which an analyst observes a vector of covariates $X$ and seeks to predict an outcome $Y$. Section \ref{sec:LLMPredict} defines a prediction rule obtained by providing the LLM with a text description of $X$ and eliciting a prediction of $Y$. This approach contrasts with the standard machine learning paradigm, reviewed in Section \ref{sec:StandardML}, in which the analyst first collects a dataset of $(x,y)$ observations and then trains a task-specific prediction algorithm on that data. Finally, Section \ref{sec:ProposedMeasure} introduces our proposed measure: the smallest training sample size for which a domain-specific model or algorithm achieves predictive performance comparable to that of the LLM.

\subsection{Setup} \label{sec:Setup}
 A \emph{prediction problem} is a pair $(X,Y)$ where $X \in \mathcal{X}$ is a covariate vector and $Y \in \mathcal{Y}$ is an outcome of interest. 

 \begin{example} The covariate vector $X$ describes a worker's educational attainment, years of experience, occupation, industry, geographic location, and past earnings. The outcome $Y \in \mathbb{R}_+$ is the worker's wage in the following year. \end{example}

 \begin{example} The covariate vector $X$ describes an individual's age, educational attainment, marital status, and race. The outcome $Y\in [0,1]$ is the probability with which the individual owns a house.
 \end{example}

A \emph{prediction rule} is any map $f: \mathcal{X} \rightarrow \mathcal{Y}$ that predicts the outcome $Y$ given $X$, where the LLM-based prediction rule is one special choice of $f$ (see Section \ref{sec:LLMPredict}).  As is standard, we measure the accuracy of a prediction $\hat{y}$ given true outcome $y$ using a loss function $\ell: \mathcal{Y} \times \mathcal{Y} \rightarrow \mathbb{R}$, where $\ell(y,y')$ is the inaccuracy of predicting $y'$ when the true outcome is $y$ (e.g., $\ell(y,y')=(y'-y)^{2}$ for regression or $\ell(y,y')=1\{y'\neq y\}$
for classification). Define the  \emph{error} of prediction rule $f$ to be
\[e(f)\equiv E_P[\ell(f(X),Y)] \]
i.e., the expected loss when predictions are made using $f$, where $P$ denotes the true (unknown) joint distribution of $(X,Y)$. This is also known as the \emph{risk} of $f$.

\subsection{LLM-Based Prediction} \label{sec:LLMPredict}

Our procedure and its guarantees generalize to any fixed prediction rule $f$, but we are in particular interested in the prediction rule $f_{\mathrm{LLM}}$ that corresponds to asking a pretrained LLM to predict the outcome. The covariate vector $X$ and desired outcome $Y$ are described in natural language and given to the LLM, and the LLM's response is taken as the predicted outcome.

Here is an example prompt:
 \footnotesize 
\begin{myprompt}{1}\label{P_ex}\texttt{Here is your persona: You are a 37-year-old male living in Tennessee in 2020. You work in architecture and engineering occupations in manufacturing industry. You have 4 years of work experience in the current job. You have\\ completed 14 years of education. You are married and have 1 child. Your father has 12 years of education. Your mother has 12 years of education. Your race is white. Your health is excellent. You are a Protestant.
\\
\\
Given your persona, do you own a house or apartment or do you rent? If you own, say 1; if you rent, say 5; if you neither own nor rent, say 8.}\end{myprompt}

\normalsize
\noindent The prompt text is fixed up to placeholders for individual covariate values, so the structure of the prompt is constant and varies only in the realized covariate values. Let
\[e_{\mathrm{LLM}} \equiv e(f_{\mathrm{LLM}})\]
denote the error of this prediction rule.

Throughout we treat the LLM-based  rule $f_{\mathrm{LLM}}$ as fixed, and do not model the background data on which it is trained. In practice, researchers often do not retrain foundation models, but instead deploy a given pretrained system. Our object of interest is therefore the predictive performance of this fixed system.

\subsection{Domain-Specific Learning} \label{sec:StandardML}

We compare the above approach to prediction rules that are obtained by training a model or algorithm on domain-specific data. This approach first specifies a set of prediction rules $\mathcal{F}$, and then uses data to select among them. Formally, an \emph{algorithm} $a$ is a map
\[a: \mathcal{D} \rightarrow \mathcal{F}\]
from the set of all finite datasets $\mathcal{D} \equiv \cup_{N \geq 1} (\mathcal{X} \times \mathcal{Y})^N$ to the set of prediction rules $\mathcal{F}$.\footnote{All of our procedures and results extend if we instead model the algorithm as a map $a: \mathcal{D} \times [0,1] \rightarrow \mathcal{F}$, where the  additional argument in $[0,1]$ represents an exogenous source of randomness, allowing the algorithm's output to be stochastic---for example, due to random initialization, subsampling, or internal randomization.}

A typical choice of $a$ would be an economic model or a machine-learning procedure. In the first case, $\mathcal{F}$ represents an interpretable class of economically motivated prediction rules, and $a(D)$ denotes the rule obtained by estimating the model's parameters on the dataset $D$ using a specified estimation method. In the second case, $\mathcal{F}$ is a potentially rich class of predictors---such as ensembles of decision trees---and $a(D)$ is the fitted predictor produced by training the algorithm on $D$, including any regularization or tuning steps. In both cases, the algorithm maps data into a concrete prediction rule whose out-of-sample performance typically improves with the size of the training dataset, providing a natural benchmark against which to evaluate the predictive value of a fixed LLM.

For any sample size $N \in \mathbb{N}_+$,  let
\[f_N \equiv a(D_N)\]
be the prediction rule obtained by training the algorithm on a random sample $D_N \equiv \{(x_i,y_i)\}_{i=1}^N$ of $N$ observations $(x_i,y_i)$ drawn i.i.d. from $P$. Since $D_N$ is random, the prediction rule $f_N$ is as well. We define 
\[e_N^a \equiv E[e(f_N)] = E_{D_N}\left[E_{(x,y)}[ \ell(a(D_N),(x,y))]\right]\]
to be the expected prediction error of the algorithm trained on $N$ datapoints. That is, $e^a_N$ integrates both over the randomness of the sample $D_N$ used to train $f_N$ and the randomness of the new observation $(X,Y)\sim P$ on which $f_N$ is evaluated.

\subsection{Proposed Measure} \label{sec:ProposedMeasure}

Our proposed measure of the LLM's \emph{equivalent sample size} is the minimum amount of domain-specific data an experimenter would need to collect so that the algorithm $a$, when trained on that data, produces predictions at least as accurate as the LLM's.

\begin{definition} Fix a prediction problem $(X,Y)$. The $a$-\emph{equivalent sample size} ($a$-ESS) of $f_{\mathrm{LLM}}$ is 
\[
N^*\equiv N^{*}_a=\min\left\{ N\in \mathbb{N}_+:e^a_N\le e_{\mathrm{LLM}} \right\} .
\]
i.e., the smallest  training sample
size $N^{*}$ such that the predictive performance of $f_{N^{*}}$
weakly improves upon that of the LLM. If no such $N$ exists, we set $N^* = \infty$.
\end{definition}

The equivalent sample size depends on the benchmark algorithm $a$, as well as on the choice of LLM, its prompting protocol, and the prediction problem. Thus it should be interpreted as measuring the LLM's predictive ability in a specific prediction problem and relative to the specified algorithm, rather than as an absolute measure of the LLM's predictive content. We encourage the use of algorithms $a$ that are well-suited to the prediction problem and thus serve as more demanding benchmarks.

At one extreme, $N^*=1$ means that the LLM provides no predictive advantage relative to the algorithm $a$ in this prediction problem, even when the algorithm is trained on a single data point. This case corresponds (for example) to an LLM that is very poorly suited to the prediction task, or which relies on spurious correlations in its original training data. 

At the other extreme, the equivalent sample size can also be infinite, meaning that no amount of domain-specific data allows the benchmark algorithm to match the LLM's predictive performance. This could occur if $a$ is too restrictive to approximate the LLM's mapping even asymptotically---for example, if $a$ is an economic model that imposes strong structural restrictions.

In general, we expect the LLM's equivalent sample size to be finite. That is, while the LLM's pretraining  enables it to make useful predictions even in the absence of any new data, a sufficiently flexible algorithm trained on data from the relevant prediction task eventually outperforms it. 

In our application we will report figures such as Figure \ref{fig:ESS}, which plots the constant error $e_{\mathrm{LLM}}$ against $e_N^a$ as $N$ varies.  The intersection point defines the $a$-equivalent sample size $N^{*}$. 

\begin{figure}
\noindent \begin{centering}
\includegraphics[scale=0.6]{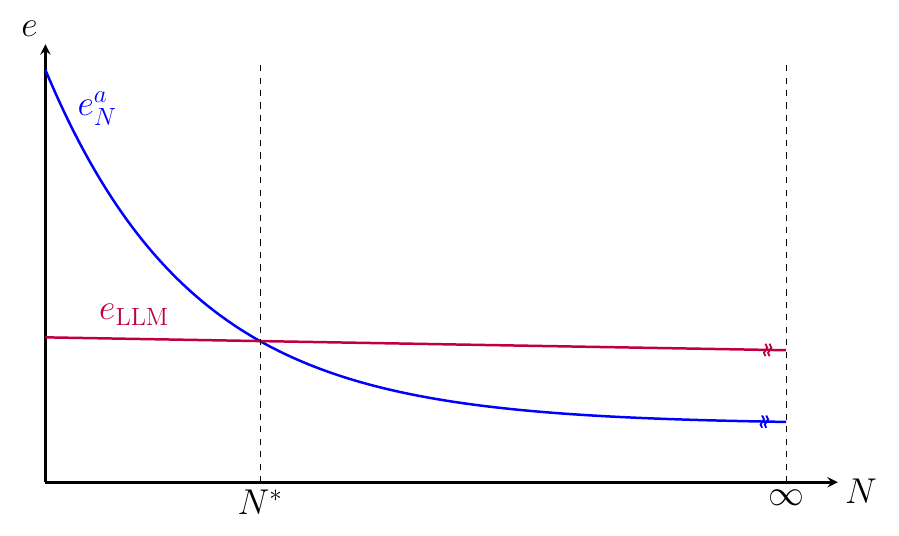}
\par\end{centering}
\caption{$a$-Equivalent Sample Size ($N^{*}$)}
\label{fig:ESS}
\end{figure}

\section{Estimation and Inference: Procedure} \label{sec:Procedure}

This section describes our procedure for estimating and conducting inference for the $a$-equivalent sample size $N^{*}$. Recall that $N^{*}$ is defined as the smallest training sample size $N$ given which  the machine learning algorithm weakly outperforms the LLM-based prediction rule:
\[
N^{*} \;\equiv\; \min\left\{N : e_N^{a} \le e_{\mathrm{LLM}}\right\}.
\]
Section \ref{sec:EstimateLLM} describes our approach to estimating the error of the LLM-based prediction rule, $e_{\mathrm{LLM}}$. Section \ref{sec:EstimateML} describes a block-out cross-validation procedure for estimating the error curve of the machine learning algorithm, $\{e_N^{a}\}_{N \ge 1}$. These can be combined to compute a plugin estimator for $N^*$, as we describe in Section \ref{sec:Plugin}. Finally, Section \ref{sec:CI} develops a sequential hypothesis-testing procedure that delivers a one-sided confidence interval (CI) for $N^*$. To aid readability, this section focuses on the procedure itself, while Section \ref{sec:InferenceTheory} provides theoretical guarantees for these procedures.

\subsection{Estimating the LLM Benchmark} \label{sec:EstimateLLM}

We estimate the error of the LLM-based prediction rule, $e_{\mathrm{LLM}}$, as its performance on the full data sample $\{(X_i,Y_i)\}_{i=1}^{n}$. For each observation $X_i$, obtain a prediction $\widehat{Y}_i^{\mathrm{LLM}} = f_{\mathrm{LLM}}(X_i)$ by querying the LLM using the fixed prompt template described in Section~\ref{sec:LLMPredict}. The empirical error of the LLM is then given by
\begin{equation} \label{eq:estimateLLM}\widehat{e}_{\mathrm{LLM}} \equiv \frac{1}{n}\sum_{i=1}^{n}\ell\!\left(\widehat{Y}_i^{\mathrm{LLM}},\,Y_i\right).
\end{equation}

\subsection{Estimating the Error Curve of the ML Algorithm} \label{sec:EstimateML}

We next estimate the algorithm's error curve by training it on disjoint training blocks of a given size and evaluating each fitted predictor on the held-out complement. Averaging the resulting out-of-sample losses produces an estimate of the algorithm's expected error as a function of the size of the training sample. Throughout this section, we drop dependence of $e_N^a$ on the algorithm $a$ to ease notation.

\begin{figure}[h]
\begin{center}
    \includegraphics[scale=0.35]{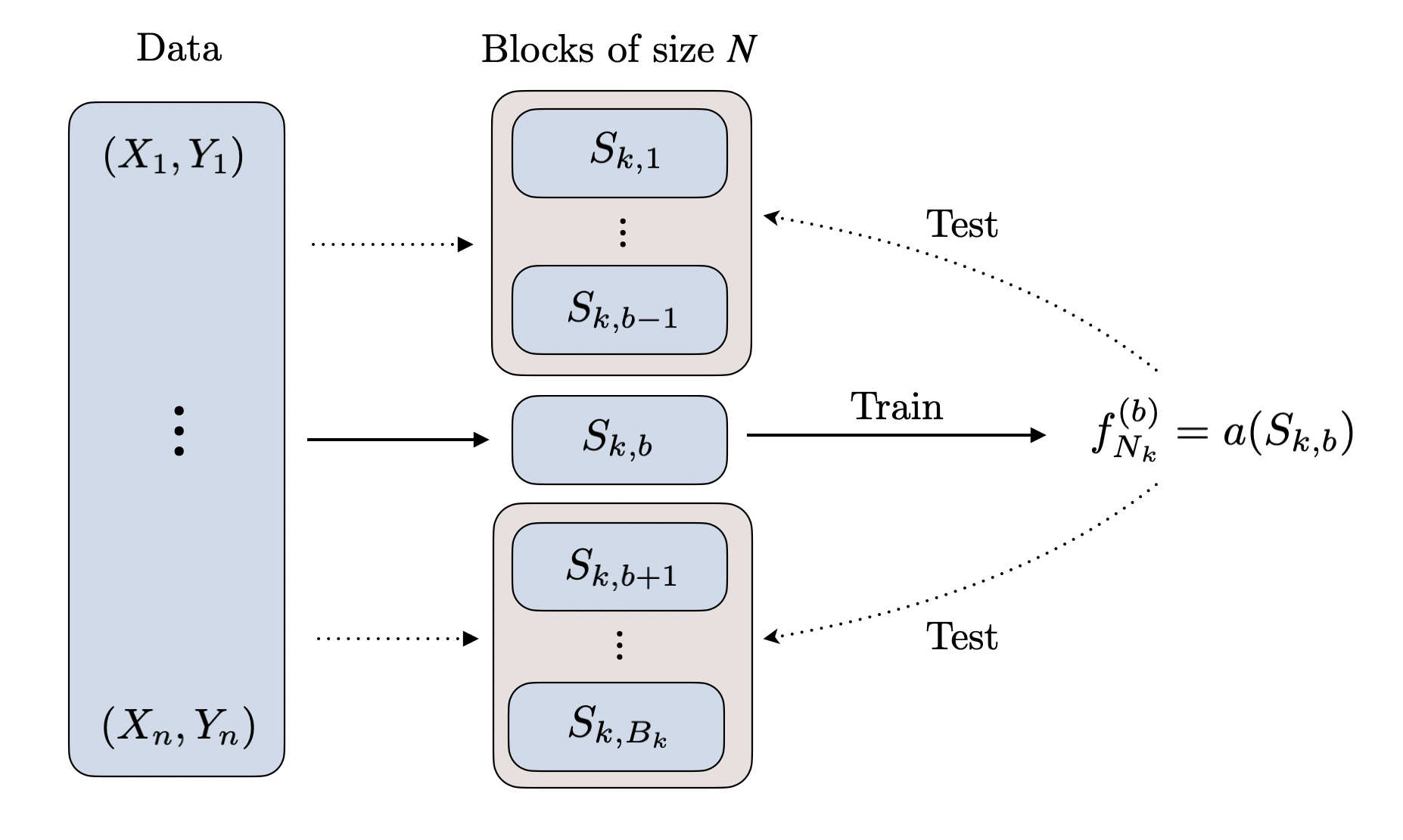}
    \caption{\small{Depiction of block-out cross-validation: The data is split into $B_k$ blocks each of size $N_k$. In each round, one block is used for training and the remaining are used for testing.}} \label{fig:blockoutCV}
    \end{center}
\end{figure}

Given a dataset of size $n$, fix a candidate training size $N_k$ for $k\in \{1,...,K\}$. Set $B_k \equiv \lfloor n / N_k \rfloor$ to be the number of blocks and $M_n \equiv n - N_k$ to be the size of the test set. For simplicity, assume $n = B_k N_k$; otherwise, discard the last $n - B_k N_k = O(1)$ observations.\footnote{Discarding these $B_kN_k$ observations does not affect the later asymptotic theory.}

Partition the index set $\{1, \ldots, n\}$ into $B_k$ disjoint blocks $S_{k,1}, \ldots, S_{k,B_k}$, each of size $N_k$. For concreteness, take $S_{k,b} = \{(b-1)N_k + 1, \ldots, b N_k\}$ for $b = 1, \ldots, B_k$. For each block $b \in \{1, \ldots, B_k\}$, define the \emph{block $b$ training set} to be \[\mathcal{D}^{(k,b)}_{\mathrm{train}} \equiv \{(X_i, Y_i) : i \in S_{k,b}\}\]
and the \emph{block $b$ test set} to be
\[\mathcal{D}^{(k,b)}_{\mathrm{test}} \equiv \{(X_i, Y_i) : i \notin S_{k,b}\}\]
as depicted in Figure \ref{fig:blockoutCV}. Because the observations $(X_i,Y_i)$ are i.i.d.\ and the sets $S_{k,b}$ form a partition of $\{1, \ldots, n\}$, the training blocks $\mathcal{D}^{(k,1)}_{\mathrm{train}}, \ldots, \mathcal{D}^{(k,B_k)}_{\mathrm{train}}$ are independent and each has distribution $P^{N_k}$.

Recall that each training set is of size $N_k$ while each test set is of size $M_n \equiv n - N_k$. In each round $b$, we fit the learning algorithm on the block $b$ training set to obtain the block-wise predictor
\[
f^{(b)}_{N_k} \equiv a\bigl(\mathcal{D}^{(k,b)}_{\mathrm{train}}\bigr),
\]
and evaluate its empirical error on the test complement:
\[
\widehat{e}_{k,b} \equiv \frac{1}{M_n} \sum_{i \notin S_{k,b}} \ell\bigl(f^{(b)}_{N_k}(X_i), Y_i\bigr).
\]
The block-out CV estimator of $e_{N_k}$ is finally
\begin{equation}
\label{eq:bCV-def}
\widehat{e}^{CV}_{N_k} \equiv \frac{1}{B_k} \sum_{b=1}^{B_k} \widehat{e}_{k,b}.
\end{equation}
i.e., the average test error across the choice of training block. 

\subsection{Plugin Estimator for the ESS} \label{sec:Plugin}

Given the estimates (\ref{eq:estimateLLM}) and (\ref{eq:bCV-def}), the plug-in estimator of the $a$-equivalent sample size is defined as
\[\widehat{N}^{*} \equiv \min\left\{N_k : \widehat{e}_{N_k}^{\mathrm{CV}} \le \widehat{e}_{\mathrm{LLM}},~k \in \{1,...,K\}\right\}
\]
i.e., the smallest sample size at which the estimated LLM error $\widehat{e}_{\mathrm{LLM}}$ exceeds the estimated error $\widehat{e}^{CV}_N$.

The next section defines a sequential estimator $\widehat N^{*}_{\alpha}$ that incorporates sampling uncertainty and delivers a lower confidence bound for $N^{*}$ with confidence level $1-\alpha$. Under standard regularity conditions, both estimators are asymptotically equivalent, but $\widehat N^{*}_{\alpha}$ is conservatively biased in finite samples to ensure correct coverage.

\subsection{One-Sided Confidence Interval for the ESS} \label{sec:CI}

To develop a one-sided CI for $N^*$, we compare the predictive performance of the machine learning algorithm and the LLM across an increasing sequence of training sample sizes, and identify the smallest training
size at which we can no longer conclude that the LLM outperforms the machine learning algorithm at the desired confidence level.

Let $\{N_{k}\}_{k=1}^{K}$ be an increasing sequence of
candidate training sizes, with $N_{1}<N_{2}<\cdots<N_{K}$. For each $k\in\{1,\dots,K\}$, consider the one-sided hypotheses
\[
H_{0,k}:e_{N_{k}}\le e_{\mathrm{LLM}},\qquad H_{1,k}:e_{N_{k}}>e_{\mathrm{LLM}}.
\]
Each hypothesis $H_{0,k}$ compares the predictive performance of the machine learning algorithm trained on $N_k$ 
 observations to that of the LLM, with the null asserting that the machine learning algorithm weakly outperforms the LLM at that training size. To test $H_{0,k}$, we use the test statistic
\begin{equation}\label{eq:t_stat}
T_{N_{k},n}\equiv\frac{\widehat{e}_{N_{k}}^{\mathrm{CV}}-\widehat{e}_{\mathrm{LLM}}}
{\widehat{\sigma}_{k}/\sqrt{n}},    
\end{equation}
where $\widehat\sigma_{k}^{2}$ is an estimator of the variance
of $\widehat e_{N_k}^{\mathrm{CV}} - \widehat e_{\mathrm{LLM}}$, as detailed in \eqref{eq:test-stat-corr} or \eqref{eq:t_stat_con} in Section \ref{subsec:Test_k} below.
We reject $H_{0,k}$
at level $\alpha$ when $T_{N_{k},n}>c_{\alpha}$ where $c_{\alpha}$ is the critical value.

We will be interested in the error curve $\{e_{N_k}\}_{k \geq 1}$. Theoretically it is possible that $e_N$ is not monotone in $N$, so that increasing the training size \emph{reduces} expected performance. This behavior is somewhat counterintuitive, and we expect that it is not descriptive of most economic models and machine learning algorithms that are considered in practice. Our subsequent results will impose the following assumption.

\begin{assumption}[Monotonicity] \label{assu:Monotone} 
$e_N$ is nonincreasing in $N$. 
\end{assumption}

Given a monotone error curve, we can conduct inference on $N^{*}$ using the following sequential testing procedure at overall level $\alpha$:
\begin{algorithm}\label{alg:seq_test_0} Proceed as follows:
\begin{list}{}{\leftmargin=1.5em \itemsep=0.5em}
\item[\textbf{Step 1.}]
     Test $H_{0,1}$ at level $\alpha$. If $H_{0,1}$ is rejected, proceed to Step 2. Otherwise, 
    stop and set $\widehat{N}_{\alpha}^{*} 
    = 1$. \\ \hspace*{2.2em}\vdots \vspace*{-1em}
  \item[\textbf{Step $k$.}] Test $H_{0,k}$ at level $\alpha$. If $H_{0,k}$  is rejected, proceed to Step $k+1$. Otherwise,
    stop and set $\widehat{N}_{\alpha}^{*} 
    =N_{k-1}+1$. \\ \hspace*{2.2em}\vdots \vspace*{-1em}
\item[\textbf{Step $K$.}] Test $H_{0,K}$ at level $\alpha$. 
    If  $H_{0,K}$  is rejected, conclude $N^{*}>N_{K}$ (and set $\widehat{N}_{\alpha}^{*} %
    =N_{K}+1$). Otherwise, set $\widehat{N}_{\alpha}^{*} %
    =N_{K-1}+1$.
    \end{list}
\end{algorithm}
The procedure sequentially increases the training sample size of the machine-learning algorithm and tests, at each step, whether its predictive performance weakly exceeds that of the LLM. As long as the null hypothesis is rejected, the procedure continues to larger training sizes. At the first step $k$ at which the corresponding null $H_{0,k}$ cannot be rejected, the procedure stops and sets the estimated equivalent sample size to $N_{k-1}+1$, the smallest sample size exceeding the largest training size for which rejection occurred.\footnote{We set $\widehat{N}_{\alpha}^{*}$ to $N_{k-1}+1$ rather than $N_k$ to account for the possible coarseness of the grid (i.e., $\{N_1,...,N_K\} \subsetneq \{1,...,K\}$). When one employ the finest grid (i.e, $N_k=k$), then simply $N_{k-1}+1 = N_k$.} If all null hypotheses up to $H_{0,K}$ are rejected, the procedure concludes that the equivalent sample size exceeds $N_K$.

 This stopping rule yields a valid lower confidence bound for the equivalent sample size $N^*$. In particular, suppose Assumption \ref{assu:Monotone} holds and the level $\alpha$ test in each step of Algorithm \ref{alg:seq_test_0} is valid. Then we show below that 
 \begin{equation}\label{eq:CI_def}
    \text{CI}_{N^*}\equiv \left[\widehat{N}_{\alpha}^{*},\infty\right)
 \end{equation}
is a one-sided $(1-\alpha)$-confidence interval for $N^{*}$. That is,
$ \Pr\!\left(\widehat{N}_{\alpha}^{*}\le N^{*}\right)\ge 1-\alpha$, i.e., with probability at least $1-\alpha$, one can conclude that the true equivalent sample size is no smaller than $\widehat{N}_\alpha^*$. In practice, this allows one to  say ``$N^{*}$ is at least as large as $\widehat{N}_{.05}^{*}$ with $95\%$ confidence,'' so that we have a minimal guarantee on the value of the LLM prediction rule.  The validity of the tests in Algorithm \ref{alg:seq_test_0} and the CI \eqref{eq:CI_def} are established in Section \ref{sec:InferenceTheory}.

\begin{remark}[Alternative CI Representation]
\label{rem:CI_duality}
The CIs for $N^*$ can be expressed in terms of CIs for the \emph{excess risk} of the ML algorithm relative to the LLM (i.e., $e^a_{N_k} - e_{\mathrm{LLM}}$). Recall that the algorithm stops at the first $N_k$ where the null hypothesis $H_{0,k}: e^a_{N_k} \le e_{\mathrm{LLM}}$ is \emph{not} rejected.
By the duality between hypothesis tests and CIs, non-rejection at level $\alpha$ corresponds to the value $0$ being contained in the one-sided CI for the difference $e^a_{N_k} - e_{\mathrm{LLM}}$.

Formally, let $LB_{k}$ denote the lower bound of the asymptotic one-sided $(1-\alpha)$ CI for the excess risk, constructed from the test statistic \eqref{eq:t_stat}.
\[
LB_{k} \equiv \left(\widehat{e}^{CV}_{N_k} - \widehat{e}_{\mathrm{LLM}}\right) - z_{1-\alpha} \frac{\widehat{\sigma}_{k}}{\sqrt{n}}.
\]
Assume the finest grid (i.e., $N_k=k$) for simplicity. The sequential estimator $\widehat{N}_{\alpha}^{*}$ is simply the smallest sample size where this lower bound is non-positive (indicating that we cannot rule out that the ML model has matched the LLM):
\begin{equation}\label{eq:CI_duality}
\widehat{N}_{\alpha}^{*} = \min \left\{ N_k : LB_{k} \le 0 \right\} = \min \left\{ N_k : \widehat{e}^{CV}_{N_k} - z_{1-\alpha} \widehat{\sigma}_{k}/\sqrt{n} \le  \widehat{e}_{\mathrm{LLM}}\right\}.\footnote{For a generic coarser grid, $\widehat{N}_{\alpha}^{*}=\max\left\{N_k : LB_k>0\right\} + 1$.}    
\end{equation}
This expression holds regardless of whether the LLM's error is treated as fixed or estimated, provided the variance estimator $\widehat{\sigma}_{k}$ correctly accounts for the variability of both terms.
\end{remark}

\section{Estimation and Inference: Theory}\label{sec:InferenceTheory}

This section presents theoretical results justifying our proposed procedure in Section \ref{sec:Procedure}. It is organized as follows: First, Section \ref{subsec:SeqTest} proves the validity of the sequential testing procedure described in Section \ref{sec:CI}, taking as given the validity of the individual tests of the hypotheses $H_{0,k}: e^a_{N_k} \le e_{\mathrm{LLM}}$. Section \ref{subsec:bCV} then establishes a central limit theorem for the block-out CV error estimator of each $e^a_{N_k}$, which yields  pointwise confidence intervals for the error curve $\{e^a_{N_k}\}_{k\geq 1}$. Finally, Section \ref{subsec:Test_k} uses this central limit theorem to establish the validity of individual test of each $H_{0,k}$.

\subsection{Validity of Sequential Hypothesis Testing}
\label{subsec:SeqTest}

The validity of Algorithm \ref{alg:seq_test_0} in terms of overall size control is established as follows.

\begin{theorem}[Validity of Sequential ESS Inference]
\label{thm:SeqTest}
Suppose that Assumption \ref{assu:Monotone} holds. 
Assume that for each $k$, the test of $H_{0,k}$ employed in Algorithm \ref{alg:seq_test_0} has asymptotic size controlled by $\alpha$, i.e.,
\begin{equation} \label{eq:singletest_valid}
 \limsup_{n \to \infty} \mathrm{Pr}(\text{reject } H_{0,k} \mid H_{0,k} \text{ is true}) \le \alpha.
\end{equation}
Then, the sequential testing procedure is valid at overall level $\alpha$ and the resulting CI for $N^*$ has a correct coverage:
\[
\liminf_{n \to \infty} \mathrm{Pr}(N^* \ge \widehat{N}_{\alpha}^{*}) \ge 1 - \alpha.
\]
\end{theorem}

The proof of Theorem \ref{thm:SeqTest} is short but reveals a structural feature of the sequential test that avoids the size distortions typical of multiple hypothesis testing. We therefore present it in the main text.

\begin{proof} 
Define
 $k^{*}\equiv\min\left\{k:e_{N_{k}}\le e_{\mathrm{LLM}}\right\}$ 
to be the index of the equivalent sample size in our sequence $\{N_1,...,N_K\}$. Note that $N_{k^*-1} + 1 \leq N^* \leq  N_{k^*}$ by the coarseness of the grid $\{N_k\}$.

Write $\underline{N}_{k} \equiv N_{k-1}+1$ with $N_0\equiv0$. We will show that $[\widehat{N}_{\alpha}^{*},\infty)$ is a valid $(1-\alpha)$-CI for $\underline{N}_{k^*}$, i.e., 
\begin{equation}\label{eq:cover_underN}
\liminf_{n \to \infty} \mathrm{Pr}(\underline{N}_{k^*} \ge \widehat{N}_{\alpha}^{*}) \ge 1 - \alpha.    
\end{equation}
Then, since $\underline{N}_{k^*} \leq N^*$  and consequently $\{\underline{N}_{k^*} \ge \widehat{N}_{\alpha}^{*}\} \subseteq \{N^* \ge \widehat{N}_{\alpha}^{*}\}$, we can conclude that  $[\widehat{N}_{\alpha}^{*},\infty)$ is also valid for $N^*$ based on
\begin{equation}\label{eq:CI_porting}
 \liminf_{n \to \infty} \mathrm{Pr}(N^* \ge \widehat{N}_{\alpha}^{*})  \ge \liminf_{n \to \infty} \mathrm{Pr}(\underline{N}_{k^*} \ge \widehat{N}_{\alpha}^{*}) \ge 1 - \alpha.   
\end{equation}

We now prove \eqref{eq:cover_underN}. 
By Assumption \ref{assu:Monotone}, the sequence of null hypotheses $H_{0,k}: e_{N_k} \le e_{\mathrm{LLM}}$ has a nested truth structure:
\begin{itemize}
    \item[](S1) For all $k < k^*$, we have $e_{N_k} > e_{\mathrm{LLM}}$, so the null $H_{0,k}$ is \emph{false}.
    \item[](S2) For all $k \ge k^*$, we have $e_{N_k} \le e_{\mathrm{LLM}}$, so the null $H_{0,k}$ is \emph{true}.
\end{itemize}
Thus, $H_{0,k^*-1}$ is the \emph{last false null hypothesis} in the sequence, while $H_{0,k^*}$ is the \emph{first true null hypothesis} in the sequence.

Now, let $\widehat{k}$ denote the random index at which Algorithm \ref{alg:seq_test_0} terminates. The CI lower bound is defined as $\widehat{N}_{\alpha}^{*}  = \underline{N}_{\widehat{k}}$. 
The algorithm stops at step $k$ if it fails to reject $H_{0,k}$. Therefore, the event $\{\widehat{k} = k\}$ implies that $H_{0,k}$ was not rejected, while all preceding hypotheses $H_{0,1}, \dots, H_{0,k-1}$ were rejected.
Consequently, the event $\{\widehat{N}_{\alpha}^{*} > \underline{N}_{k^*}\} = \{\underline{N}_{\widehat{k}} > \underline{N}_{k^*}\}$ is equivalent to $\{\widehat{k} > k^*\}$.

For the algorithm to go beyond step $k^*$ (i.e., for $\widehat{k} > k^*$), it must have rejected the hypothesis $H_{0,k^*}$. Conversely, if $H_{0,k^*}$ is rejected, then $\widehat{k} > k^*$. 
Therefore,
\[
\Pr( \widehat{N}_{\alpha}^{*} > \underline{N}_{k^*} ) = \Pr\left( \text{reject } H_{0,k^*} \right).
\]
Since $H_{0,k^*}$ is true by (S2), the rejection of $H_{0,k^*}$ constitutes a Type I error (a false rejection).
By the condition of the theorem, the test for each individual step controls the asymptotic size at level $\alpha$. Specifically, for the step $k^*$, we have $\limsup_{n \to \infty} \Pr\left( \text{reject } H_{0,k^*} \right) \le \alpha$. 

Combining these results, we obtain the bound on the coverage probability:
\begin{align*}
\liminf_{n \to \infty} \Pr\left( N^* \ge \widehat{N}_{\alpha}^{*} \right) 
&= 1 - \limsup_{n \to \infty} \Pr\left( \widehat{N}_{\alpha}^{*} > N^{*} \right) \\
&\ge 1 - \limsup_{n \to \infty} \Pr\left( \text{reject } H_{0,k^*} \right) 
\ \ge 1 - \alpha.
\end{align*}
which implies the desired conclusion in view of \eqref{eq:CI_porting}.
\end{proof}

An important feature of this procedure is that it controls the family-wise error rate (FWER) (i.e., the probability of rejecting any true nulls) without requiring Bonferroni-type corrections \citep[e.g.,][]{holm1979simple}. This is because the sequence of hypotheses are ordered (nested) and the sequential testing procedure stops at the first non-rejection. Thus we only test true null hypotheses if we have correctly rejected all preceding false nulls, so the probability of making \emph{any} false rejection is bounded simply by the probability of rejecting the \emph{first} true null $H_{0,k^*}$. This property is particularly valuable given that we will typically want to evaluate a wide range of sample sizes.

\begin{remark}[Impact of the Coarse Grid $\{N_k\}_{k=1}^K$] Let $\widehat{N}^*_{\alpha,\mathrm{fine}}$ denote the cutoff $\widehat{N}_\alpha^*$ obtained under the ``finest grid'' $N_k = k$. Then $\widehat{N}^*_{\alpha,\mathrm{fine}}$ coincides with the smallest integer $k$ at which we reject the null in Algorithm \ref{alg:seq_test_0}, which would produce a ``sharp'' $(1-\alpha)$ CI for $N^*$ in the form of $[\widehat{N}^*_{\alpha,\mathrm{fine}},\infty)$. For a generic coarse grid $\{N_k\}$, the cutoff $\widehat{N}^*_\alpha$ is conservative against LLM in the sense of $\widehat{N}^*_{\alpha} \leq \widehat{N}^*_{\alpha,\mathrm{fine}}$. Equivalently, we may view $[\widehat{N}^*_{\alpha,\mathrm{fine}},\infty)$ as a sharp CI for $\underline{N}_{k^*} \leq N^*$, which results in conservative one-sided coverage for $N^*$.
\end{remark}

The validity of our sequential test, as established in Theorem \ref{thm:SeqTest}, requires a valid test statistic at each step.  Recall in Section \ref{sec:CI} that we use the studentized test statistic \eqref{eq:t_stat} for the boundary of the null hypothesis, $\tilde H_{0,k}: e_{N_k} = e_{\mathrm{LLM}}$, i.e., the null where the machine learning algorithm and the LLM have equal expected error. Since this is the least favorable configuration within the null---i.e., the case where false rejection is most
likely---a level-$\alpha$ test of $\tilde H_{0,k}$ is also a valid level-$\alpha$ test of the original null hypothesis $H_{0,k}$.
 An asymptotic one-sided test rejects $H_{0,k}$ at level $\alpha$ if $T_{N_{k},n} > z_{1-\alpha}$, where $z_{1-\alpha}$ is the $(1-\alpha)$-quantile of the standard normal distribution.\footnote{We reject the null (that ML is as good as LLM) when the ML error is significantly \emph{higher} than the LLM error. Thus, the rejection region is the right tail.}
The validity of this individual test follows from Proposition \ref{prop:test_validity_corr} in Section \ref{subsec:Test_k}, which is based on  the Central Limit Theorem and consistent variance estimation in Section \ref{subsec:bCV}.

\subsection{Central Limit Theorem under Block-Out Cross Validation}
\label{subsec:bCV}

This section develops a Central Limit Theorem for risk estimators using the block-out CV procedure described in Section \ref{sec:EstimateML}. Our established CLT in Theorem \ref{thm:CLT-bCV}, and the corresponding consistent variance estimator in Theorem \ref{thm:CLT_varest} are established under very mild conditions without the requirement of stability conditions \citep{lei2025modern}, and may be of independent interest beyond the scope of this paper.

Specifically, we consider a setting where researchers have access to a single large data set with sample size $n$. For each candidate training size $N_k$, we estimate the out-of-sample error $e^a_{N_k}$ of any learning algorithm $a$ via block-out CV. We work with an asymptotic framework where $n \to \infty$ while $N_k$ remains fixed.\footnote{In Appendix \ref{sec:regimeB}, we consider an alternative asymptotic framework where $N_k \to \infty$ proportionally to $n$ (i.e., $B$ is fixed), and establish the validity of our approach under an alternative set of assumptions.} These CIs in turn produce CIs for $N^*$ as discussed in Remark \ref{rem:CI_duality}.

We work under the following sampling framework and regularity conditions. Let $Z \equiv (X, Y) \sim P$ denote the distribution of observed data. 

\begin{assumption}[Random Sampling]
\label{assu:sampling} A random sample of data $D_n\equiv\left(Z_1,...,Z_n\right)$ is observed, with 
$Z_i \equiv (X_i, Y_i) \sim_\text{i.i.d.} Z\equiv (X,Y) \sim P$.
\end{assumption}

For each given $N\leq n$ and a sample $D_N$ of size $N$, we define a generic learning algorithm as a measurable map $a: \mathcal{Z}^N \to \mathcal{F}$, where $\mathcal{F}$ is a class of measurable functions $f: \mathcal{X} \to \mathcal{Y}$. For any training sample $D_N \in \mathcal{Z}^N$, recall that we write $f_{N} = a(D_N)$. 

We consider a generic loss function $\ell$ that satisfies the following conditions:
\begin{assumption}[Loss and Moments]
\label{assu:moments}
The loss function $\ell: \mathcal{Y} \times \mathcal{Y} \to \mathbb{R}$ is measurable and, for $D_N \sim P^N$ and an independent $Z = (X, Y) \sim P$,
\[
\mathbb{E}\bigl[\,|\ell(f_{N}(X), Y)|^{2+\delta}\bigr] < \infty \quad \text{for some } \delta > 0.
\]
\end{assumption}

Recall that $e(f_{N}) \equiv \mathbb{E}_Z\bigl[\ell(f_{N}(X), Y) \mid D_N,\bigr]$ and $e^a_N \equiv \mathbb{E}_{D_N}[e(f_{N})]$. In addition, for each realization $z = (x, y) \in \mathcal{Z}$, define
\[
e^a_N(z) \equiv \mathbb{E}_{D_N}\bigl[\ell(f_{D_N}(x), y)\bigr],
\]
where the expectation is over an independent copy $D_N \sim P^N$ with $z$ fixed. Note that $e^a_N = \mathbb{E}_Z[e^a_N(Z)]$. Under Assumption \ref{assu:moments}, $e(f_{N})$, $e^a_N(z)$, and $e^a_N$ are well-defined with $\mathbb{E}\bigl[e(f_{N})^2\bigr],\  \mathbb{E}\bigl[e^a_N(Z)^2\bigr] < \infty$. 
We now present our main result.

\begin{theorem}[CLT for block-out CV Risk Estimator]
\label{thm:CLT-bCV}
Suppose Assumptions \ref{assu:sampling}-\ref{assu:moments} hold and $N_k$ is fixed while $n \to \infty$. Then
\[
\sqrt{n}\bigl(\widehat{e}^{CV}_{N_k} - e^a_{N_k}\bigr) \xrightarrow{d} \mathcal{N}\bigl(0, \sigma^2_{N_k}\bigr),
\]
where the asymptotic variance $\sigma^2_{N_k}$ is given by
\begin{equation}
\label{eq:asym-var}
\sigma^2_{N_k} \equiv N_k V_{N_k,\mathrm{train}} + V_{N_k,\mathrm{test}} + 2 N_k C_{N_k},
\end{equation}
with $V_{N_k,\mathrm{train}} \equiv \mathrm{Var}\bigl(e(f_{N_k})\bigr)$, $V_{N_k,\mathrm{test}} \equiv \mathrm{Var}\bigl(e^a_{N_k}(Z)\bigr)$, and $C_{N_k} \equiv \mathrm{Cov}\bigl(e(f_{N_k}), e^a_{N_k}(Z_1)\bigr)$, where $Z_1 \in \mathcal{D}_{N_k}$.
\end{theorem}

Importantly, we note that the asymptotic variance of the block-out CV estimator features three components: $V_{N_k, \text{train}}$ that captures the training-sample randomness contribution, $V_{N_k, \text{test}}$ that captures the testing-sample randomness contribution, and $C_{N_k}$ that reflects the fact that each observation appears both in one training block and in the test sets of all other blocks, a key feature of cross-validation approaches.

\subsubsection{Asymptotic Variance Estimation} \label{sec:AsympVar}

We now construct an estimator $\widehat{\sigma}^2_{N_k}$ of the asymptotic variance $\sigma^2_{N_k}$ by plugging in estimators for each component of the three-term decomposition in \eqref{eq:asym-var}, as described below. 

\paragraph{Estimator of $V_{N_k,\text{train}}$}
For each block $b$, define the block-wise population risk $e_{b,\mathrm{train}} \equiv \mathbb{E}\bigl[\ell(f^{(b)}_{N_k}(X), Y) \mid \mathcal{D}^{(k,b)}_{\mathrm{train}}\bigr]$. The $(e_{b,\mathrm{train}})_{b=1}^{B_k}$ are i.i.d.\ with mean $e^a_{N_k}$ and variance $V_{N_k,\mathrm{train}}$. We estimate $V_{N_k,\mathrm{train}}$ by the sample variance of the block-wise CV risks:
\[
\widehat{V}_{N_k,\mathrm{train}} \equiv \frac{1}{B_k - 1} \sum_{b=1}^{B_k} \bigl(\widehat{e}_{k,b} - \widehat{e}^{CV}_{N_k}\bigr)^2.
\]

\paragraph{Estimator of $V_{N_k,\text{test}}$}
Each $i$ belongs to exactly one training block and appears in the test sets of all other $B_k - 1$ blocks. Define the block-averaged test loss for observation $i$:
\[
\widehat{\mu}_i \equiv \frac{1}{B_k - 1} \sum_{b : i \notin S_{k,b}} \ell\bigl(f^{(b)}_{N_k}(X_i), Y_i\bigr), \qquad i = 1, \ldots, n.
\]
As $B_k \to \infty$, $\widehat{\mu}_i \xrightarrow{p} e^a_{N_k}(Z_i)$. We estimate $V_{N_k,\mathrm{test}}$ by:
\[
\widehat{V}_{N_k,\mathrm{test}} \equiv \frac{1}{n-1} \sum_{i=1}^{n} \bigl(\widehat{\mu}_i - \bar{\mu}\bigr)^2, \qquad \bar{\mu} \equiv \frac{1}{n} \sum_{i=1}^{n} \widehat{\mu}_i.
\]

\paragraph{Estimator of $C_{N_k}$}
For each block $b$, define the block-wise summaries
\[
\widehat{e}_b \equiv \widehat{e}_{k,b} , \qquad \widehat{m}_b \equiv \frac{1}{N_k} \sum_{i \in S_{k,b}} \widehat{\mu}_i, \qquad b = 1, \ldots, B_k.
\]
Let $\bar{r} \equiv B_k^{-1} \sum_{b=1}^{B_k} \widehat{e}_b$ and $\bar{m} \equiv B_k^{-1} \sum_{b=1}^{B_k} \widehat{m}_b$. We estimate the cross term by:
\[
\widehat{C}_{N_k} \equiv \frac{1}{B_k - 1} \sum_{b=1}^{B_k} (\widehat{e}_b - \bar{r})(\widehat{m}_b - \bar{m}).
\]

Combining the above, we construct the asymptotic variance estimator as
\begin{equation}
\label{eq:plugin-var}
\widehat{\sigma}^2_{N_k} \equiv N_k \widehat{V}_{N_k,\mathrm{train}} + \widehat{V}_{N_k,\mathrm{test}} + 2 N_k \widehat{C}_{N_k},
\end{equation}
whose validity is established in the following theorem.

\begin{theorem}[Studentized CLT]
\label{thm:CLT_varest}
Suppose Assumptions \ref{assu:sampling} and \ref{assu:moments} hold and $N_k$ is fixed while $n \to \infty$. Let $\widehat{\sigma}^2_{N_k}$ be defined by \eqref{eq:plugin-var}. Then:
\begin{enumerate}
\item $\widehat{\sigma}^2_{N_k} \xrightarrow{p} \sigma^2_{N_k}$ as $n \to \infty$.
\item If $\sigma^2_{N_k} > 0$, then the studentized statistic satisfies
\[
\frac{\sqrt{n}\bigl(\widehat{e}^{CV}_{N_k} - e^a_{N_k}\bigr)}{\widehat{\sigma}_{N_k}} \xrightarrow{d} \mathcal{N}(0, 1).
\]
\end{enumerate}
\end{theorem}

\subsection{A Valid Test of $H_{0,k}: e^a_{N_k} \le e_{\mathrm{LLM}}$}\label{subsec:Test_k}

We finally show how to construct a valid hypothesis test of $H_{0,k}: e^a_{N_k} \le e_{\mathrm{LLM}}$ based on Theorem \ref{thm:CLT_varest}, which can then be used to implement the sequential testing procedure in Algorithm \ref{alg:seq_test_0}.

We base our test on the difference between the estimated risks:
\[
\widehat{\Delta}_k \equiv \widehat{e}^{CV}_{N_k} - \widehat{e}_{\mathrm{LLM}}.
\]
Since both terms are computed on the same dataset, they are correlated.
To account for this correlation within the fixed-$N_k$ asymptotic framework, we apply the general block-out CV theory directly to the \emph{loss difference}.
Define the loss difference for a given predictor $f$ and observation $Z=(X,Y)$ as 
$$\delta(f, Z) \equiv \ell(f(X), Y) - \ell(f_{\mathrm{LLM}}(X), Y).$$
Note that the estimator $\widehat{\Delta}_k$ is algebraically equivalent to the block-out CV estimator applied to this difference loss function.\footnote{Note that $\widehat{e}_{\mathrm{LLM}}$ is the sample mean of the LLM loss. Since the LLM prediction rule is fixed, this is identical to the block-out CV estimator for a fixed function $f_{\mathrm{LLM}}$.} Applying Theorem \ref{thm:CLT_varest} to the loss difference $\delta$, we obtain the following result.

\begin{proposition}[Test Validity]
\label{prop:test_validity_corr}
Define the studentized test statistic:
\begin{equation}
\label{eq:test-stat-corr}
T_{\mathrm{diff},N_k,n} \equiv \frac{\widehat{e}^{CV}_{N_k} - \widehat{e}_{\mathrm{LLM}}}{\widehat{\sigma}_{\mathrm{diff}, k}/\sqrt{n}}.
\end{equation}
where $\widehat{\sigma}^2_{\mathrm{diff}, k}$ estimates the asymptotic variance of the difference. Under Assumptions \ref{assu:sampling} and \ref{assu:moments}, under the boundary of the null hypothesis ($e^a_{N_k} = e_{\mathrm{LLM}}$):
\[
T_{\mathrm{diff},N_k,n} \xrightarrow{d} \mathcal{N}(0, 1) \quad \text{as } n \to \infty.
\]
provided that $\sigma^2_{\mathrm{diff},k} >0$.  
Consequently, the one-sided test that rejects $H_{0,k}$ when $T_{N_k,n} > z_{1-\alpha}$ has asymptotic size $\alpha$.
\end{proposition}

\begin{remark}[Conservative Alternative]
In practice, implementing the exact difference variance estimator $\widehat{\sigma}^2_{\mathrm{diff}, k}$ requires computing covariances between ML and LLM errors.
A simpler alternative is to use the sum of their individual variances, i.e.,
\begin{equation}\label{eq:t_stat_con}
    T_{\mathrm{con},N_k,n} \equiv \frac{\widehat{e}^{CV}_{N_k} - \widehat{e}_{\mathrm{LLM}}}{\sqrt{\left(\widehat{\sigma}^2_{N_k} + \widehat{\sigma}^2_{\mathrm{LLM}}\right)/n}}.
\end{equation}
Note that $\sigma^2_{\mathrm{diff}, k} = \sigma^2_{N_k} + \sigma^2_{\mathrm{LLM}} - 2\mathrm{Cov}(\sqrt{n}\widehat{e}^{CV}_{N_k}, \sqrt{n}\widehat{e}_{\mathrm{LLM}})$.
Assuming that the predictions of the ML model and the LLM are positively correlated (which is typical, as ``hard'' observations tend to yield high losses for both models), the covariance term is positive.
In this case, the sum of variances overestimates the true variance ($\widehat{\sigma}^2_{N_k} + \widehat{\sigma}^2_{\mathrm{LLM}} > \widehat{\sigma}^2_{\mathrm{diff}, k}$), making the test statistic $T_{N_k,n}$ smaller and the test conservative (size $< \alpha$).
\end{remark}

\section{Applications} \label{sec:Application}

\subsection{Data and Prediction Problems}

The Panel Study of Income Dynamics (PSID) is one of the longest-running
longitudinal household surveys in the United States, following individuals
and families since 1968. It collects detailed information on economic,
social, and health outcomes and is widely used in empirical research and
policy analysis. We use the 2021 PSID wave and  evaluate LLM prediction of  four outcomes $Y$: hourly wage ($Y \in \mathbb{R}_+$), homeownership ($Y \in \{\mbox{own},\mbox{rent}\}$), smoking $(Y \in \{0,1\}$), and drinking $(Y \in \{0,1\})$. In each prediction problem, the covariate vector $X$ denotes the respondent's  characteristics, including age,
sex, race, education, occupation, hours worked, labor market experience,
self-reported health, state of residence, marital status, number of children,
parents' education, religion, and the other  outcomes listed above except for the one being predicted.

For the continuous-valued outcome (hourly wage), we set the loss function to be squared-error loss $\ell(y,y') = (y' - y)^2$, although for readability we report the \emph{root mean-squared error}
\[\mathrm{RMSE}(f) = \sqrt{e(f)}\]
rather than the error $e(f)$.
For the other discrete outcomes (homeownership, drinking, and smoking) we set the usual 0-1 loss $\ell(y,y') = 1\{y' \neq y\},$
so error is the misclassification rate.

\subsection{LLM Prediction Rules}\label{subsec:application_LLM}

We generate LLM predictions as follows.  For each test observation
$i=1,\ldots,n$, we transform the covariate vector $X_i$
into a natural-language \emph{persona} description and then prompt the LLM to
predict $Y_i$. Importantly, we do not fine-tune the LLM on PSID data and do
not use few-shot examples drawn from PSID. Recall Prompt \ref{P_ex}.

A natural concern is potential data leakage \citep{ludwig2025large}, namely that PSID data could have been included---directly or indirectly---in
LLMs' training corpora. While access to PSID data requires registration and
human verification through the PSID Data Center, we nonetheless implement two
safeguards: (i) we evaluate a static LLM
with a documented training data descriptions and cutoffs:  OpenAI's ChatGPT-4 (cutoff: April 2023),\footnote{We access these models via an API. Based on interviews with current and former OpenAI employees, we expect that when ChatGPT is accessed via the API it is not subject to further fine-tuning (e.g., through reinforcement learning from human feedback), and that the reported knowledge cutoff is reliable.} and (ii) we use the 2021 wave of PSID released on October 2024,\footnote{\hyperlink{https://www.icpsr.umich.edu/web/sbeccc/studies/39190/versions/V1}{https://www.icpsr.umich.edu/web/sbeccc/studies/39190/versions/V1}} which postdates the reported knowledge
cutoffs of the LLMs. (See Appendix \ref{sec:PSID_more} for supplementary results considering additional LLMs whose cutoffs also predate the 2021 PSID release.)

\subsection{Domain-Specific Learning}

We compare the LLM's predictions to those of  supervised ML models
trained directly on PSID data. 
We consider two popular ML algorithms, which are briefly reviewed here. 

\paragraph{Lasso regression and Logit L1.} Lasso is a regularized linear prediction method. Different from standard OLS, it seeks to maximize fit to the data subject to a penalty for model complexity, evaluated as the sum of the absolute value of coefficients. Lasso regression thus encourages small coefficients, and often leads to sparse linear models in which many coefficients are set equal to zero. Logit L1 applies the same $\ell_1$-regularization principle in a logistic regression framework for binary outcomes.

\paragraph{Random forests.} Random forests are a flexible prediction method built by averaging many  decision trees. Each component tree is trained on a different subsample of the data and uses a different subset of covariates, and the final prediction aggregates across them. This procedure allows the algorithm to capture nonlinear relationships and interactions without requiring these features to be specified in advance. 

\bigskip

For each ML model, we compute the error curve, namely, the test error as a function of the training sample size $N$. We tune hyperparameters (e.g., the lasso penalty parameter, the depth and leaf size for RF) for every $N$ using $k$-fold CV; see Section \ref{sec:ML-detail} in the Appendix for the implementation details.
 These error curve estimates are subsequently compared to the LLM benchmark to compute the ESS.

\subsection{Results}

We estimate error curves for each prediction
problem and compute the equivalent sample size of the LLM. Figures \ref{fig:learning_curve1}--\ref{fig:learning_curve4} present the full estimated error curves, the  estimated equivalent sample size $\widehat{N}^*$, and  confidence intervals for all estimated quantities.  The reported confidence intervals for the block-out CV errors are calculated as hybrid CIs: they are based on the fixed-$N$ asymptotic theory (Theorem \ref{thm:CLT-bCV}) when $N\le400$ and the fixed-$B$ asymptotic theory (Corollary \ref{corol:CLT-regimeB} in Section \ref{sec:regimeB}) otherwise.\footnote{This choice is made based on calibrated simulation to ensure desirable coverage performance.} The confidence intervals for $N^*$ are calculated based on \eqref{eq:CI_duality} in Remark \ref{rem:CI_duality}. Table \ref{tab:ess_by_outcome} more succinctly presents our confidence intervals for  the equivalent sample size. 

\begin{table}[!htbp]
\centering
\caption{Equivalent sample size across prediction tasks} 
\label{tab:ess_by_outcome}

\begin{threeparttable}
\begin{adjustbox}{max width=\linewidth}
\begin{tabular}{lccc}
\toprule
Outcome & {LLM Error} & Algorithm & One-Sided CI\\
\midrule

Hourly Wage & 12.50  & Lasso     & $[20,\ \infty)$ \\
            &          & RF      & $[17,\ \infty)$ \\
\addlinespace

Home Ownership   & 0.24 & Logit L1       & $[676,\ \infty)$ \\
            &          & RF           & $[590,\ \infty)$ \\

\addlinespace

Drinking    & 0.4      & Logit L1 & $[59,\ \infty)$ \\
            &          & RF       & $[41,\ \infty)$ \\

\addlinespace

Smoking     & 0.23    & Logit L1      & $[1,\ \infty)$ \\
            &          & RF            & $[1,\ \infty)$ \\

\bottomrule
\end{tabular}
\end{adjustbox}

\end{threeparttable}
\end{table}

We find substantial heterogeneity in the LLM's predictive performance across outcomes. The LLM performs extremely well in predicting homeownership: at the 95\% confidence level, its equivalent sample size is at least 590 observations when benchmarked against a random forest and at least 676 observations when benchmarked against L1-penalized logit.

The LLM's performance is weaker, but still meaningful, for predicting drinking behavior. At the same confidence level, the LLM's equivalent sample size is at least 41 observations relative to the random forest and at least 59 observations relative to L1-penalized logit.

Predictive performance deteriorates further for hourly wages. Here, we can only conclude that the LLM's equivalent sample size is at least 17 observations when compared to the random forest and at least 20 observations when compared to L1-penalized logit.

Finally, the LLM provides essentially no predictive value for smoking behavior: we fail to reject the null hypothesis that conventional machine-learning algorithms outperform the LLM even when trained on a single observation.

Taken together, these results illustrate that the value of LLM predictions is highly task-specific---ranging from a close substitute for hundreds of observations in some settings to providing no meaningful predictive content in others. This underscores the importance of evaluating the LLM's predictive performance across domains.

\begin{figure}
\noindent \begin{centering}
\includegraphics[scale=0.6]{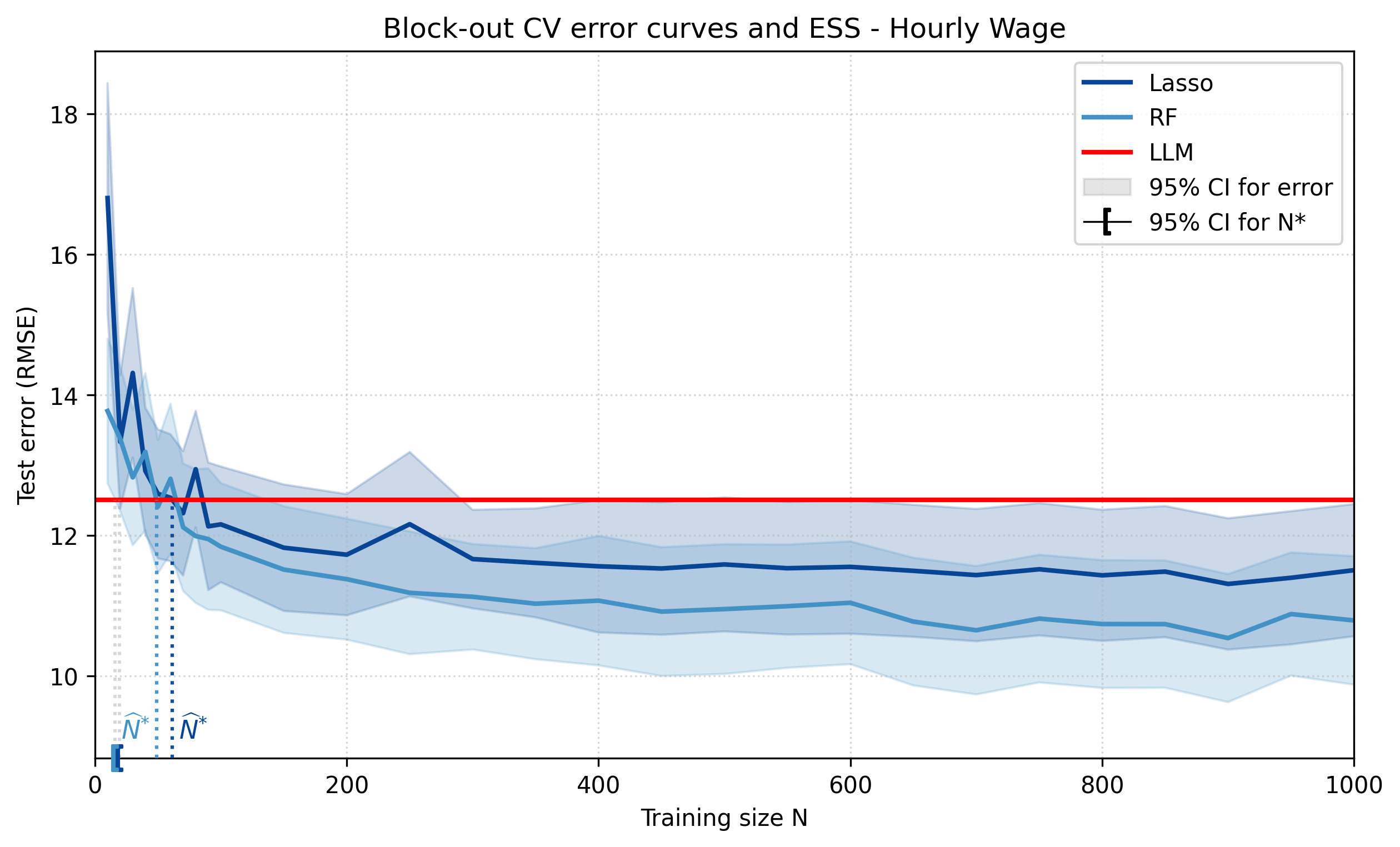}
\par\end{centering}
\caption{RMSE across training sample sizes ($Y$: Hourly wage)}
\label{fig:learning_curve1}
\end{figure}

\begin{figure}
\noindent \begin{centering}
\includegraphics[scale=0.6]{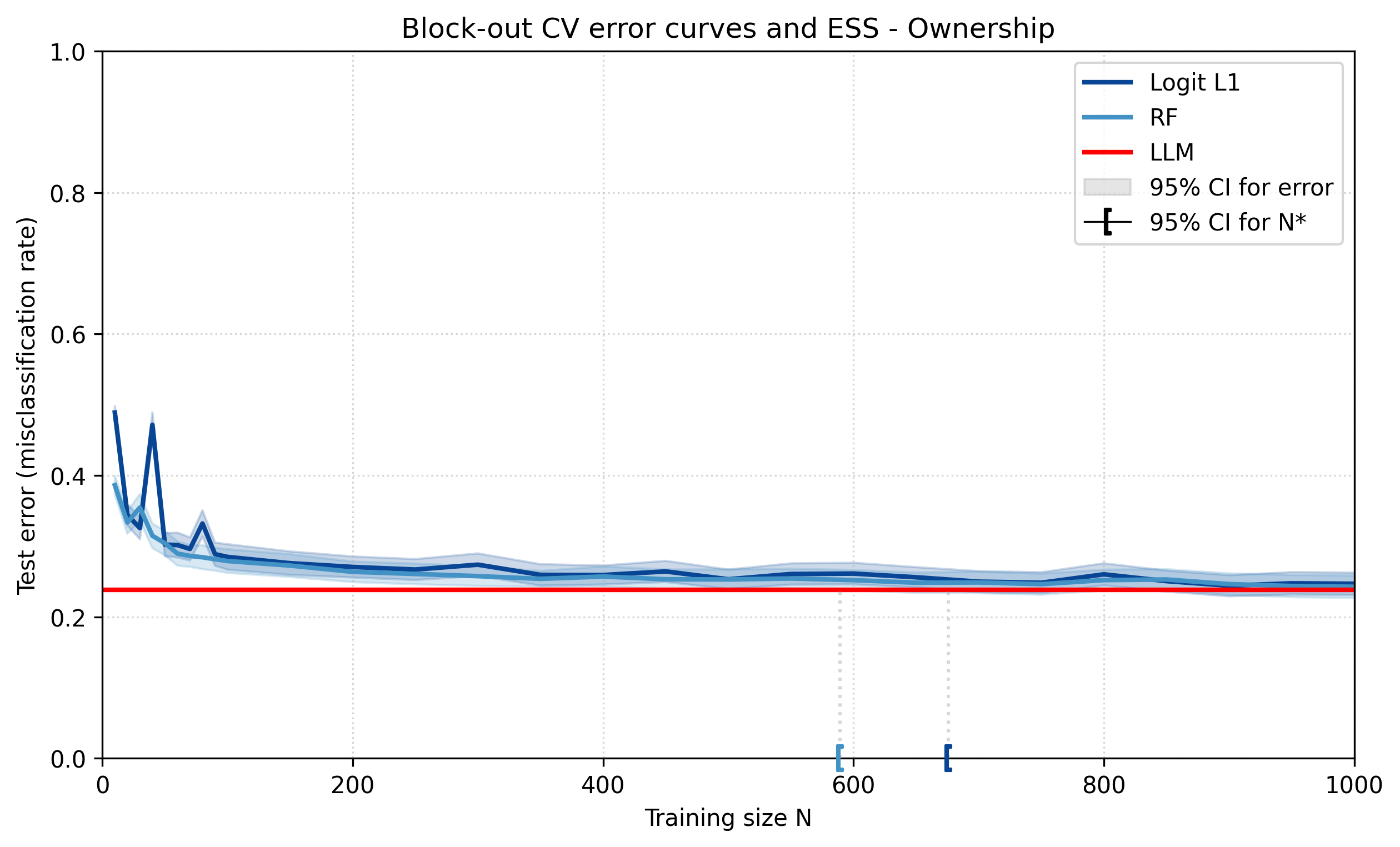}
\par\end{centering}
\caption{Misclassification rate across training sample sizes ($Y$:
Homeownership)}
\label{fig:learning_curve3}
\end{figure}

\begin{figure}
\noindent \begin{centering}
\includegraphics[scale=0.6]{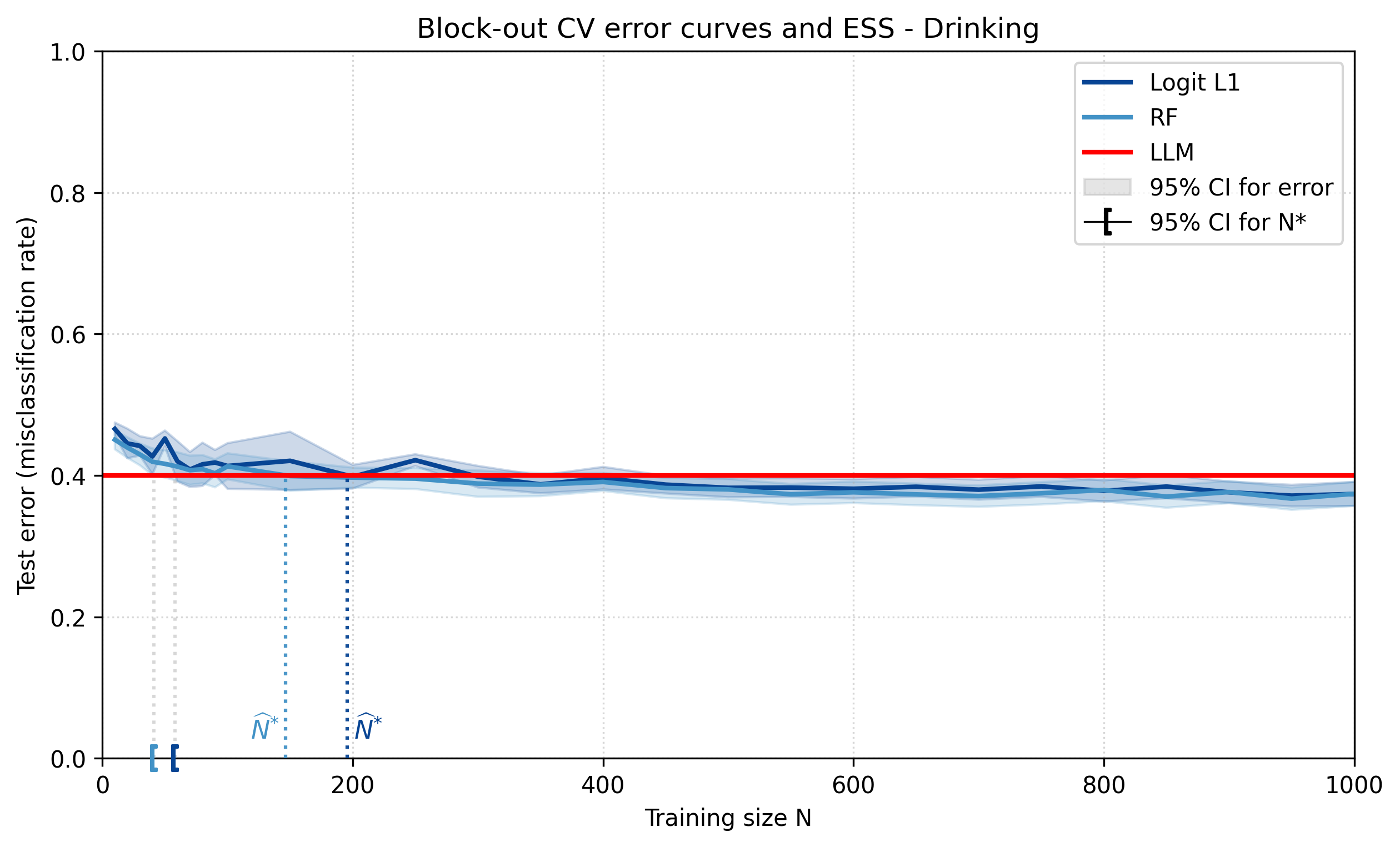}
\par\end{centering}
\caption{Misclassification rate across training sample sizes ($Y$:
Drinking)}
\label{fig:learning_curve2}
\end{figure}

\begin{figure}
\noindent \begin{centering}
\includegraphics[scale=0.6]{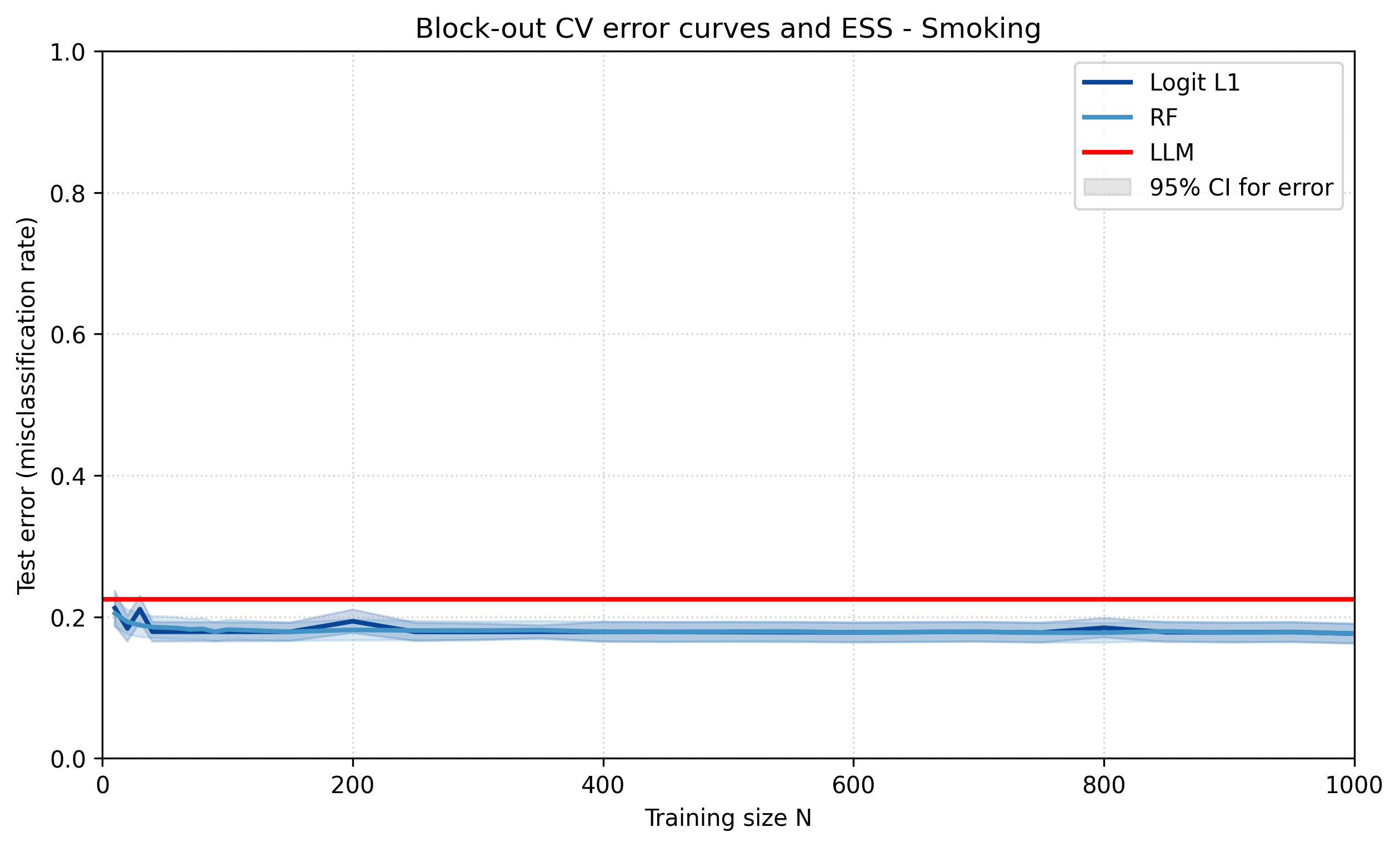}
\par\end{centering}
\caption{Misclassification rate across training sample sizes ($Y$:
Smoking)}
\label{fig:learning_curve4}
\end{figure}

\section{Extension: ESS for Treatment Effects}\label{sec:CATE}

LLMs are considered  useful tools for causal inference \citep{kiciman2023causal,wang2024twin,han2025mining, liu2025large}. In light of this line of work, we extend our evaluation framework for LLMs that are employed for treatment effect estimation. 

Let $\mu(X)\equiv \mathbb{E}_{P}[Y\mid X]$. Under squared loss $\ell(y,\hat y)=(y-\hat y)^2$, the population risk admits the standard bias--variance-type decomposition
\begin{align*}
e(f)\equiv \mathbb{E}_{P}\!\left[(Y-f(X))^{2}\right]
&=\mathbb{E}_{P}\!\left[(Y-\mu(X))^{2}\right]+\mathbb{E}_{P}\!\left[(\mu(X)-f(X))^{2}\right]\\
&\equiv e_{\text{irreducible}}+e_{\mu}(f),
\end{align*}
where the equality holds as the cross term is zero by the law of iterated expectations and the irreducible error is achieved as $\mu(X)$ uniquely minimizes $\mathbb{E}_{P}[(Y-f(X))^{2}]$. Therefore, any ESS defined with respect to $e(f)$ would be equivalent to an ESS defined with respect to $e_{\mu}(f)$.

Now consider the estimation of the conditional average treatment effect (CATE) with data $(X,T,Y)$ distributed as $P$, where $T$ is an indicator for an additional treatment. Assume unconfoundedness: $T\perp (Y_{0},Y_{1})\mid X$ where $Y_t$ is the counterfactual outcome for $t\in\{0,1\}$. Then,
the conditional mean of $Y_t$ is identified as $\mu_{t}(X)\equiv \mathbb{E}_{P}[Y_{t}\mid X]=\mathbb{E}_{P}[Y\mid X,T=t]$. Within treatment arm $d$ and under squared loss $\ell(y,\hat y)=(y-\hat y)^2$, the conditional risk admits the standard bias--variance-type decomposition:
\begin{align*}
e_t(f) &\equiv \mathbb{E}_{P}\!\left[(Y_t-f(X))^{2}\right]\\
&=\mathbb{E}_{P}\!\left[(Y_t-\mu_{t}(X))^{2}\right]+\mathbb{E}_{P}\!\left[(\mu_{t}(X)-f(X))^{2}\right]\\
&\equiv e_{\text{irreducible},t}+e_{\mu_t}(f),
\end{align*}
by an argument analogous to above. Therefore, any ESS defined with respect to $e_t(f)$ would be equivalently an ESS defined with respect to $e_{\mu_t}(f)$. Moreover, $e_t(f)$ can be computed from observed data for each $T=t$, as $\mathbb{E}_{P}\!\left[(Y_t-f(X))^{2}\right]  =\mathbb{E}_{P} \left[\mathbb{E}_{P}\!\left[(Y-f(X))^{2}\mid X, T=t\right]\right]$ by unconfoundedness.

For the LLM, we can estimate $e_t(f_{\text{LLM}})$ by applying the prompt strategy in Section \ref{sec:LLMPredict} to the test sample restricted to $T=t$ (i.e., to the treated or control test sample).\footnote{This approach relates to counterfactual generation using LLMs in clinical trials \citep{wang2024twin}.} For ML trained on each sample of size $N$, define
\[
e_{N,t}^{a}\equiv \mathbb{E}_{D_{N}}\, e_t(f_{N})=\mathbb{E}_{D_{N}}\, e_t(a(D_{N})),
\qquad D_{N}\equiv\{(x_{i},t_{i},y_{i})\}_{i=1}^{N},
\]
where $e_{N,t}^{a}$ can be evaluated via block-out cross-validation within the $T=t$ subsample. Then, for $t\in\{0,1\}$ the treatment-specific ESS is defined as
\[
N_{t}^{*}\equiv \min\{N:\; e_{N,t}^{a}\le e_t(f_{\mathrm{LLM}})\}.
\]

Define the CATE $\tau(X)\equiv \mathbb{E}_{P}[Y_{1}-Y_{0}\mid X]$ and the CATE risk
\[
e_{\tau}(g)\equiv \mathbb{E}_{P}[(\tau(X)-g(X))^{2}].
\]
For both the LLM and ML, using prediction rules $(f_{1},f_{0})$ corresponding to $(e_{\mu_1}(f_{1}),e_{\mu_0}(f_{0}))$ only yields the upper bound on $e_{\tau}(f_{1}-f_{0})$ due to the nonlinearity of the loss function, and thus defining ESS for $\tau$ through $f_{1}-f_{0}$ is elusive.\footnote{Specifically, we only have $e_{\tau}(f_{1}-f_{0})
=\mathbb{E}_{P}\!\left[(\tau(X)-(f_{1}(X)-f_{0}(X)))^{2}\right]
\le 2\{e_{\mu_1}(f_{1})+e_{\mu_0}(f_{0})\}$.} We thus proceed as follows. Instead of prompting the LLM to predict $Y_t$, we prompt it to directly ``predict'' $\tau(X)$. An example prompt can be as follows:
\footnotesize 
\begin{myprompt}{2}\label{P_ex_CATE}\texttt{Here is your persona: You are a 37-year-old male living in Tennessee in 2020. You work in architecture and engineering occupations in manufacturing industry. You have 4 years of work experience in the current job. You have completed 14 years of education. You are married and have 1 child. Your father has 12 years of education. Your mother has 12 years of education. Your race is white. Your health is excellent. You are Protestant. 
\\
\\
Finally, you are a subject in a randomized controlled trial in which the treatment is the provision of health insurance, specifically Medicaid.
\\
\\
Given your persona, what would be the effect of the randomized provision of health insurance on the number of your primary care visits?}\end{myprompt}

\normalsize

Let $g_{\text{LLM}}(x)$ denote such a prediction. For the benchmark algorithm, let $\tilde a:\mathcal{D}\to\mathcal{G}$ map data to a CATE function $g\in \mathcal{G}$ and define
\[
e_{N,\tau}^{\tilde a}\equiv \mathbb{E}_{D_{N}}\, e_{\tau}(\tilde a(D_{N})).
\]
Then define the CATE-ESS as
\[
N_{\tau}^{*}\equiv \min\{N:\; e_{N,\tau}^{\tilde a}\le e_{\tau}(g_{\mathrm{LLM}})\}.
\]
In practice, because $\tau$ is unknown, we proceed via the transformed outcome. Under unconfoundedness, it can be shown that\footnote{Observe that \begin{align*}
E[\tilde{Y} | X] &= P(T=1|X) E[\tilde{Y} | X, T=1] + P(T=0|X) E[\tilde{Y} | X, T=0] \\
&= \pi(X) \left[ \frac{E[Y_1|X]}{\pi(X)} \right] + (1 - \pi(X)) \left[ \frac{-E[Y_0|X]}{1 - \pi(X)} \right]= E[Y_1|X] - E[Y_0|X],
\end{align*}where unconfoundedness is applied in the second equality.}
\[
\mathbb{E}_{P}[\tilde{Y}\mid X]=\tau(X),
\]
where
\[
\tilde{Y}\equiv Y\frac{T-\pi(X)}{\pi(X)(1-\pi(X))},\qquad \pi(X)\equiv \mathbb{P}[T=1\mid X].
\]
Thus,
\begin{align*}
    e_{\tau}(g)=\mathbb{E}_{P}[(\tau(X)-g(X))^{2}]
&=\mathbb{E}_{P}[(\tilde{Y}-g(X))^{2}]-\mathbb{E}_{P}[(\tilde{Y}-\tau(X))^{2}]\\
&=\tilde e(g)-\tilde e_{\text{irreducible}},    
\end{align*}
where $\tilde e_{\text{irreducible}}$ does not depend on $g$. Consequently the ESS does not change if we replace $e_{\tau}$ by $\tilde e$:
\[
N_{\tau}^{*}=\tilde N^{*}\equiv \min\{N:\; \tilde e_{N}^{\tilde a}\le \tilde e(g_{\mathrm{LLM}})\},
\qquad \tilde e_{N}^{\tilde a}\equiv \mathbb{E}_{D_{N}}\, \tilde e(\tilde a(D_{N})).
\]
In practice, $\tilde e(g_{N})$ can be estimated by computing $\tilde Y$ (e.g., with known $\pi(X)$ in a randomized control trial), and the ML model can simply predict $\tilde Y$ based on $X$.

\section{Conclusions}

LLMs have become popular prediction tools in empirical social science. We
measure the value of LLMs' pretrained knowledge using the \emph{equivalent
sample size} (ESS): the training sample size required for an algorithm trained on domain-specific data
to match the LLM benchmark performance. We propose an inference framework for
ESS based on sequential hypothesis tests and a new asymptotic theory for
cross-validation error under block-out CV. The framework can help researchers
quantify and communicate the capability (or limitation) of LLMs in
domain-specific prediction tasks.

\begin{appendix}

\section{Block-Out CV under Fixed-$N$ Asymptotics}
\label{app:block_CV}
This appendix contains proofs of the main asymptotic results stated in Section~\ref{subsec:bCV}. Throughout, we fix the training size $N_k$ and let $n \to \infty$. For notational simplicity, we suppress the subscript $k$ and write $N$ for $N_k$, $B$ for $B_k$, $S_b$ for $S_{k,b}$, etc.

Recall the following definitions. The block-out CV estimator is
\[
\widehat{e}^{CV}_N = \frac{1}{B} \sum_{b=1}^{B} \widehat{e}_b, \qquad \widehat{e}_b = \frac{1}{M_n} \sum_{i \notin S_b} \ell\bigl(f_N^{(b)}(X_i), Y_i\bigr),
\]
where $M_n = n - N$ and $f_N^{(b)} = a(\mathcal{D}^{(b)}_{\mathrm{train}})$. For each block $b$, define the block-wise error
\[
e_{b,\mathrm{train}} \equiv \mathbb{E}\bigl[\ell(f_N^{(b)}(X), Y) \mid \mathcal{D}^{(b)}_{\mathrm{train}}\bigr].
\]
By construction, $(e_{b,\mathrm{train}})_{b=1}^B$ are i.i.d.\ with $\mathbb{E}[e_{b,\mathrm{train}}] = e^a_N$ and $\mathrm{Var}(e_{b,\mathrm{train}}) = V_{N,\mathrm{train}}$.

~

\begin{proof}[\textbf{Proof of Theorem \ref{thm:CLT-bCV}}]
Decompose the CV estimator as
\[
\widehat{e}^{\mathrm{CV}} = A_n + B_n, \qquad A_n \equiv \frac{1}{B} \sum_{b=1}^{B} e_{b,\mathrm{train}}, \quad B_n \equiv \frac{1}{B} \sum_{b=1}^{B} (\widehat{e}_b - e_{b,\mathrm{train}}).
\]

\medskip
First, we consider the training contribution $A_n$. The training blocks $\mathcal{D}^{(b)}_{\mathrm{train}} = (Z_{(b-1)N+1}, \ldots, Z_{bN})$ for $b = 1, \ldots, B$ are i.i.d, and thus $e_{b,\mathrm{train}}$ are i.i.d.\ with
\[
\mathbb{E}[e_{b,\mathrm{train}}] = \mathbb{E}_{D_N}[e(f_{N})] = e^a_N, \qquad \mathrm{Var}(e_{b,\mathrm{train}}) = V_{N,\mathrm{train}}.
\]
By the classical CLT,
\[
\sqrt{B}(A_n - e^a_N) \xrightarrow{d} \mathcal{N}(0, V_{N,\mathrm{train}}).
\]
Since $B = n/N$ under our assumption that $n = BN$, we have
\[
\sqrt{n}(A_n - e^a_N) = \sqrt{N} \cdot \sqrt{B}(A_n - e^a_N) \xrightarrow{d} \mathcal{N}(0, N V_{N,\mathrm{train}}).
\]

\medskip
Second, we consider the testing contribution $B_n$. Consider the centered loss at $z \equiv (x, y)$,
\[
\psi(D_N, z) \equiv \ell(f_{D_N}(x), y) - e(f_{N}), 
\]
Then for each block $b$ and each $i \notin S_b$,
\[
\ell(f_N^{(b)}(X_i), Y_i) - e_{b,\mathrm{train}} = \psi(\mathcal{D}^{(b)}_{\mathrm{train}}, Z_i).
\]
Thus
\[
B_n = \frac{1}{B M_n} \sum_{b=1}^{B} \sum_{i \notin S_b} \psi(\mathcal{D}^{(b)}_{\mathrm{train}}, Z_i).
\]

For each pair $(b, i)$ with $i \notin S_b$, the training set $\mathcal{D}^{(b)}_{\mathrm{train}}$ and the test point $Z_i$ are independent. Therefore, the projection onto the test point is
\[
\mathbb{E}\bigl[\psi(\mathcal{D}^{(b)}_{\mathrm{train}}, Z_i) \mid Z_i\bigr] = e^a_N(Z_i) - e^a_N.
\]
In the meanwhile, the projection onto a training point $Z_j$ (where $j \in S_b$) involves taking the expectation over the test point $Z_i$ conditional on the training set, which equals zero by the construction of the centered kernel $\psi$: $\mathbb{E}[\psi(\mathcal{D}^{(b)}_{\mathrm{train}}, Z_i) \mid \mathcal{D}^{(b)}_{\mathrm{train}}] = 0$.

We define $T_n$ as the first-order H\'{a}jek projection of $B_n$:
\[
T_n \equiv \sum_{j=1}^n \mathbb{E}[B_n \mid Z_j].
\]
Since each observation $Z_j$ serves as a test point for exactly $B-1$ blocks (specifically, all blocks $b$ such that $j \notin S_b$), we obtain:
\[
T_n = \sum_{j=1}^n \frac{B-1}{B M_n} \bigl(e^a_N(Z_j) - e^a_N\bigr) = \frac{B-1}{B M_n} \sum_{i=1}^{n} \bigl(e^a_N(Z_i) - e^a_N\bigr).
\]
Define the remainder $\mathcal{R}_n \equiv B_n - T_n$. Clearly, 
\[
\frac{B-1}{B M_n} = \frac{1}{n} + O(n^{-2}),
\]
with the implicit constant depending only on $N$.

The remainder $\mathcal{R}_n$ is a higher-order term in the sense of U/V-statistics.
Note that the term $B_n$ constitutes an incomplete V-statistic of order $N+1$ with kernel 
\[
h_N(z_1, \ldots, z_N, z_{N+1}) \equiv \psi\bigl((z_1, \ldots, z_N), z_{N+1}\bigr).
\]
and $T_n$ is the first-order H\'{a}jek projection of $B_n$. The remainder $\mathcal{R}_n \equiv B_n - T_n$ represents the degenerate component of the statistic.
By construction of the projection, $\mathcal{R}_n$ is uncorrelated with any linear function of the observations, and $\mathbb{E}[\mathcal{R}_n] = 0$.
To bound its variance, we observe that $\mathcal{R}_n$ can be expressed as a sum of degenerate kernels $\eta(\cdot)$ over the block structure:
\[
\mathcal{R}_n = \frac{1}{B M_n} \sum_{b=1}^{B} \sum_{i \notin S_b} \eta(\mathcal{D}^{(b)}_{\mathrm{train}}, Z_i),
\]
where $\eta$ satisfies the degeneracy properties $\mathbb{E}[\eta \mid Z_i] = 0$ (orthogonal to test point) and $\mathbb{E}[\eta \mid \mathcal{D}^{(b)}_{\mathrm{train}}] = 0$ (orthogonal to training set).
Consequently, in the expansion of the second moment $\mathbb{E}[\mathcal{R}_n^2]$, cross-terms $\mathbb{E}[\eta(\mathcal{D}^{(b)}_{\mathrm{train}}, Z_i)\eta(\mathcal{D}^{(b')}_{\mathrm{train}}, Z_j)]$ vanish if the two terms share the same block ($b=b'$ but $i \neq j$) due to the second degeneracy property. Cross-terms with disjoint blocks ($b \neq b'$) also vanish unless the indices overlap in a specific way. The only non-vanishing terms arise from configurations where training and test indices overlap across different blocks. Given the disjoint block design, the number of such non-vanishing terms scales as $O(n^2)$, while the normalization factor is $(B M_n)^{-2} = O(n^{-2})$.
This implies the variance scaling:
\[
\mathbb{E}[\mathcal{R}_n^2] = O(n^{-2}).
\]
By Chebyshev's inequality, for any $\epsilon > 0$, $\mathbb{P}(|\sqrt{n}\mathcal{R}_n| > \epsilon) \leq \frac{n \mathbb{E}[\mathcal{R}_n^2]}{\epsilon^2} \to 0$, ensuring that $\sqrt{n}\mathcal{R}_n \xrightarrow{p} 0$.
\medskip
Third, we consider the covariance between $A_n$ and $B_n$. Since $\sqrt{n} \mathcal{R}_n \xrightarrow{p} 0$ and $\sqrt{n}(A_n - e^a_N) = O_p(1)$, we have
\[
\mathrm{Cov}\bigl(\sqrt{n}(A_n - e^a_N), \sqrt{n} \mathcal{R}_n\bigr) \to 0.
\]
Hence, it suffices to study the covariance between $A_n$ and $T_n$.
Using the definition of $T_n$,
\[
\mathrm{Cov}(A_n, T_n) = \frac{B-1}{B^2 M_n} \sum_{b=1}^{B} \sum_{i=1}^{n} \mathrm{Cov}\bigl(e_{b,\mathrm{train}}, e^a_N(Z_i)\bigr).
\]
If $i \notin S_b$, then $e_{b,\mathrm{train}}$ depends only on $\{Z_j : j \in S_b\}$ whereas $e^a_N(Z_i)$ depends only on $Z_i$ and an independent training sample, so $e_{b,\mathrm{train}}$ and $e^a_N(Z_i)$ are independent and their covariance is 0. If $i \in S_b$, then the pair $(e_{b,\mathrm{train}}, e^a_N(Z_i))$ has the same distribution as $(e(f_{N}), e^a_N(Z_1))$ with $D_N = (Z_1, \ldots, Z_N) \sim P^N$.
Hence, 
\[
\mathrm{Cov}\bigl(e_{b,\mathrm{train}}, e^a_N(Z_i)\bigr) = C_N \quad \text{for all } b \text{ and } i \in S_b.
\]
Observing that there are exactly $BN = n$ pairs $(b, i)$ with $i \in S_b$, we deduce
\[
\mathrm{Cov}(A_n, T_n) = \frac{B-1}{B^2 M_n} \cdot BN \cdot C_N = \frac{B-1}{B} \cdot \frac{N}{M_n} \cdot C_N = \frac{N}{n} C_N.
\]
and consequently
\[
\mathrm{Cov}\bigl(\sqrt{n}(A_n - e^a_N), \sqrt{n} T_n\bigr) = n \, \mathrm{Cov}(A_n, T_n) \to N C_N.
\]

\medskip
Third, we consider the covariance between $A_n$ and $B_n$. Since $\sqrt{n} \mathcal{R}_n \xrightarrow{p} 0$ and $\sqrt{n}(A_n - e^a_N) = O_p(1)$, we have
\[
\mathrm{Cov}\bigl(\sqrt{n}(A_n - e^a_N), \sqrt{n} \mathcal{R}_n\bigr) \to 0.
\]
Hence, it suffices to study the covariance between $A_n$ and $T_n$. Using the definition of $T_n$,
\[
\mathrm{Cov}(A_n, T_n) = \frac{B-1}{B^2 M_n} \sum_{b=1}^{B} \sum_{i=1}^{n} \mathrm{Cov}\bigl(e_{b,\mathrm{train}}, e^a_N(Z_i)\bigr).
\]
If $i \notin S_b$, then $e_{b,\mathrm{train}}$ depends only on $\{Z_j : j \in S_b\}$ whereas $e^a_N(Z_i)$ depends only on $Z_i$ and an independent training sample, so $e_{b,\mathrm{train}}$ and $e^a_N(Z_i)$ are independent and their covariance is 0. If $i \in S_b$, then the pair $(e_{b,\mathrm{train}}, e^a_N(Z_i))$ has the same distribution as $(e(f_{N}), e^a_N(Z_1))$ with $D_N = (Z_1, \ldots, Z_N) \sim P^N$. Hence, 
\[
\mathrm{Cov}\bigl(e_{b,\mathrm{train}}, e^a_N(Z_i)\bigr) = C_N \quad \text{for all } b \text{ and } i \in S_b.
\]
Observing that there are exactly $BN = n$ pairs $(b, i)$ with $i \in S_b$, we deduce
\[
\mathrm{Cov}(A_n, T_n) = \frac{B-1}{B^2 M_n} \cdot BN \cdot C_N = \frac{B-1}{B} \cdot \frac{N}{M_n} \cdot C_N = \frac{N}{n} C_N + O(n^{-2}).
\]
and consequently
\[
\mathrm{Cov}\bigl(\sqrt{n}(A_n - e^a_N), \sqrt{n} T_n\bigr) = n \, \mathrm{Cov}(A_n, T_n) \to N C_N.
\]

\medskip
Finally, combining the analysis above, we deduce
\[
\sqrt{n}(A_n - e^a_N) \xrightarrow{d} \mathcal{N}(0, N V_{N,\mathrm{train}}), \qquad \sqrt{n} B_n \xrightarrow{d} \mathcal{N}(0, V_{N,\mathrm{test}}),
\]
and the joint limit is bivariate normal with covariance $\lim_n \mathrm{Cov}(\sqrt{n}(A_n - e^a_N), \sqrt{n} B_n) = N C_N$. Thus
\[
\sqrt{n}(\widehat{e}^{\mathrm{CV}} - e^a_N) = \sqrt{n}(A_n - e^a_N) + \sqrt{n} B_n \xrightarrow{d} \mathcal{N}(0, \sigma_N^2),
\]
with 
$\sigma_N^2 = N V_{N,\mathrm{train}} + V_{N,\mathrm{test}} + 2 N C_N.$
\end{proof}

The proof of Theorem~\ref{thm:CLT_varest} relies on showing that each component of the plug-in variance estimator is consistent. We establish this through three lemmas.

\begin{lemma}
\label{lem:train-var-est}
Under Assumptions \ref{assu:sampling},  \ref{assu:moments}, and  $n = BN$,
\[
\widehat{V}_{N,\mathrm{train}} \xrightarrow{p} V_{N,\mathrm{train}}.
\]
\end{lemma}

\begin{proof}[\textbf{Proof of Lemma \ref{lem:train-var-est}}]
Write $\widehat{e}_b = e_{b,\mathrm{train}} + U_b$ where
\[
U_b \equiv \frac{1}{M_n} \sum_{i \notin S_b} Y_i, \quad \text{with} \quad Y_i \equiv \ell(f_N^{(b)}(X_i), Y_i) - e_{b,\mathrm{train}}.
\]
Conditional on the training set $\mathcal{D}^{(b)}_{\mathrm{train}}$, the variables $\{Y_i\}_{i \notin S_b}$ are independent with mean zero. By Assumption~\ref{assu:moments}, they satisfy $\mathbb{E}[|Y_i|^{2+\delta}] < \infty$ for some $\delta > 0$.

We recall Rosenthal's inequality \citep{rosenthal1970subspaces}: Let $X_1, \ldots, X_m$ be independent zero-mean random variables with finite $p$-th moments for $p \ge 2$. Then there exists a constant $C_p$ depending only on $p$ such that
\[
\mathbb{E}\left[\left|\sum_{j=1}^m X_j\right|^p\right] \le C_p \left( \sum_{j=1}^m \mathbb{E}[|X_j|^p] + \left( \sum_{j=1}^m \mathbb{E}[X_j^2] \right)^{p/2} \right).
\]
Applying this to $U_b$ with $p = 2+\delta$ and $m = M_n$, and noting that $\sum \mathbb{E}[|Y_i|^{2+\delta}] = M_n \mathbb{E}[|Y_1|^{2+\delta}]$ and $\sum \mathbb{E}[Y_i^2] = M_n \mathrm{Var}(Y_1)$, we obtain:
\begin{align*}
\mathbb{E}\bigl[|U_b|^{2+\delta} \mid \mathcal{D}^{(b)}_{\mathrm{train}}\bigr] 
&= \frac{1}{M_n^{2+\delta}} \mathbb{E}\left[\left|\sum_{i \notin S_b} Y_i\right|^{2+\delta} \,\bigg|\, \mathcal{D}^{(b)}_{\mathrm{train}}\right] \\
&\le \frac{C_{2+\delta}}{M_n^{2+\delta}} \left( M_n \mathbb{E}[|Y_1|^{2+\delta}] + \bigl( M_n \mathrm{Var}(Y_1) \bigr)^{1+\delta/2} \right) \\
&= O\bigl(M_n^{-(1+\delta)} + M_n^{-(1+\delta/2)}\bigr) \\
&= O\bigl(n^{-(1+\delta/2)}\bigr).
\end{align*}
To bound the maximum error across blocks, we apply the union bound and Markov's inequality:
\[
\Pr\left(\max_{1 \le b \le B} |U_b| > \epsilon\right) \le \sum_{b=1}^B \frac{\mathbb{E}[|U_b|^{2+\delta}]}{\epsilon^{2+\delta}} \le B \cdot \frac{O(n^{-(1+\delta/2)})}{\epsilon^{2+\delta}}.
\]
Since $B = O(n)$, the probability is bounded by $O(n^{-\delta/2})$, which converges to zero as $n \to \infty$. Thus, $\max_b |U_b| \xrightarrow{p} 0$.

Finally, note that $(e_{b,\mathrm{train}})$ are i.i.d.\ with variance $V_{N,\mathrm{train}}$, so the sample variance $$S_B^2(e_{1,\mathrm{train}}, \ldots, e_{B,\mathrm{train}}) \xrightarrow{p} V_{N,\mathrm{train}}$$ 
by the standard LLN. Since the sample variance $S_B^2$ is continuous in $e_b$'s and $\max_b |\widehat{e}_b - e_{b,\mathrm{train}}| = \max_b |U_b| \xrightarrow{p} 0$, we have
\[
S_B^2(\widehat{e}_1, \ldots, \widehat{e}_B) - S_B^2(e_{1,\mathrm{train}}, \ldots, e_{B,\mathrm{train}}) \xrightarrow{p} 0
\]
and thus $S_B^2(\widehat{e}_1, \ldots, \widehat{e}_B) \xrightarrow{p} V_{N,\mathrm{train}}$.
\end{proof}

\begin{lemma}
\label{lem:test-var-est}
Under Assumptions \ref{assu:sampling}, \ref{assu:moments}, and $n = BN$,
\[
\widehat{V}_{N,\mathrm{test}} \xrightarrow{p} V_{N,\mathrm{test}}.
\]
\end{lemma}

\begin{proof}[\textbf{Proof of Lemma \ref{lem:test-var-est}}]
Recall that $\widehat{\mu}_i = \frac{1}{B-1} \sum_{b : i \notin S_b} \ell(f_N^{(b)}(X_i), Y_i)$. Conditional on $Z_i$, the summands are independent and identically distributed with mean $e^a_N(Z_i)$. Let $v_N(z) \equiv \mathrm{Var}(\ell(f_N(x), y))$ denote the conditional variance of the loss given a test point $z=(x,y)$. The conditional variance of the average is:
\[
\mathrm{Var}(\widehat{\mu}_i \mid Z_i) = \frac{v_N(Z_i)}{B-1}.
\]
By the law of total expectation, the mean squared error is
\[
\mathbb{E}\bigl[(\widehat{\mu}_i - e^a_N(Z_i))^2\bigr] = \mathbb{E}\bigl[\mathrm{Var}(\widehat{\mu}_i \mid Z_i)\bigr] = \frac{\mathbb{E}[v_N(Z)]}{B-1}.
\]
Under Assumption~\ref{assu:moments}, $\mathbb{E}[v_N(Z)]$ is finite. Since $B \asymp n$ in the fixed-$N$ regime, the right-hand side converges to zero uniformly in $i$ as $n \to \infty$. Thus, $\widehat{\mu}_i$ converges to $e^a_N(Z_i)$ in $L^2$.

To establish the consistency of the sample second moment, we use the decomposition:
\[
\frac{1}{n} \sum_{i=1}^n \widehat{\mu}_i^2 = \frac{1}{n} \sum_{i=1}^n e^a_N(Z_i)^2 + \frac{1}{n} \sum_{i=1}^n \left( \widehat{\mu}_i^2 - e^a_N(Z_i)^2 \right).
\]
The first term is an average of i.i.d.\ variables with finite mean (by Assumption~\ref{assu:moments}). By the standard LLN, it converges in probability to $\mathbb{E}[e^a_N(Z)^2]$.

For the second term, we apply the Cauchy-Schwarz inequality to bound the $L^1$ norm of the difference:
\begin{align*}
\mathbb{E}\bigl| \widehat{\mu}_i^2 - e^a_N(Z_i)^2 \bigr| &= \mathbb{E}\bigl| (\widehat{\mu}_i - e^a_N(Z_i))(\widehat{\mu}_i + e^a_N(Z_i)) \bigr| \\
&\le \left\| \widehat{\mu}_i - e^a_N(Z_i) \right\|_2 \cdot \left\| \widehat{\mu}_i + e^a_N(Z_i) \right\|_2.
\end{align*}
We established above that $\| \widehat{\mu}_i - e^a_N(Z_i) \|_2 \to 0$. The second factor is bounded because $\| \widehat{\mu}_i \|_2 \to \| e^a_N(Z_i) \|_2 < \infty$. Thus, the expected absolute difference converges to zero. By Markov's inequality, the average of the differences converges to zero in probability.

Combining these results:
\[
\frac{1}{n} \sum_{i=1}^{n} \widehat{\mu}_i^2 \xrightarrow{p} \mathbb{E}[e^a_N(Z)^2], \qquad \frac{1}{n} \sum_{i=1}^{n} \widehat{\mu}_i \xrightarrow{p} \mathbb{E}[e^a_N(Z)] = e^a_N.
\]
Consequently,
\[
\widehat{V}_{N,\mathrm{test}} = \frac{n}{n-1} \left( \frac{1}{n} \sum_{i=1}^n \widehat{\mu}_i^2 - \bar{\mu}^2 \right) \xrightarrow{p} \mathbb{E}[e^a_N(Z)^2] - (e^a_N)^2 = V_{N,\mathrm{test}}.
\]
\end{proof}

\begin{lemma}
\label{lem:cross-est}
Under Assumptions~\ref{assu:sampling}, \ref{assu:moments}, and $n = BN$,
\[
\widehat{C}_N \xrightarrow{p} C_N.
\]
\end{lemma}

\begin{proof}[\textbf{Proof of Lemma \ref{lem:cross-est}}]
Define the errors $\Delta e_b = \widehat{e}_b - e_{b,\mathrm{train}}$ and $\Delta m_b = \widehat{m}_b - \widetilde{m}_b$.
From the proof of Lemma \ref{lem:train-var-est}, we established the uniform bound $\max_b |\Delta e_b| \xrightarrow{p} 0$. Consequently, the mean squared error for the training component vanishes:
\[
\frac{1}{B} \sum_{b=1}^B (\Delta e_b)^2 \le \left(\max_{b} |\Delta e_b|\right)^2 \xrightarrow{p} 0.
\]
For the test component $\widehat{m}_b$, we rely on the $L^2$ convergence from Lemma \ref{lem:test-var-est}. Note that $\widehat{m}_b - \widetilde{m}_b = \frac{1}{N} \sum_{i \in S_b} (\widehat{\mu}_i - e^a_N(Z_i))$. By Jensen's inequality and the fact that $|S_b|=N$ is fixed:
\[
(\Delta m_b)^2 = \left( \frac{1}{N} \sum_{i \in S_b} (\widehat{\mu}_i - e^a_N(Z_i)) \right)^2 \le \frac{1}{N} \sum_{i \in S_b} (\widehat{\mu}_i - e^a_N(Z_i))^2.
\]
Averaging over all $B$ blocks, we obtain
\[
\frac{1}{B} \sum_{b=1}^B (\Delta m_b)^2 \le \frac{1}{BN} \sum_{b=1}^B \sum_{i \in S_b} (\widehat{\mu}_i - e^a_N(Z_i))^2 = \frac{1}{n} \sum_{i=1}^n (\widehat{\mu}_i - e^a_N(Z_i))^2.
\]
which, by Lemma \ref{lem:test-var-est}, converges to 0 in probability. Thus, $\frac{1}{B} \sum_{b=1}^B (\Delta m_b)^2 \xrightarrow{p} 0$.

Now, express the difference in sample covariances (ignoring the $B/(B-1)$ correction which approaches 1) as:
\begin{align*}
\widehat{C}_N - \widetilde{C}_N &\approx \frac{1}{B} \sum_{b=1}^B \left( \widehat{e}_b \widehat{m}_b - e_{b,\mathrm{train}} \widetilde{m}_b \right) - \bar{\widehat{e}}\bar{\widehat{m}} + \bar{e}_{\mathrm{train}}\bar{\widetilde{m}} \\
&= \frac{1}{B} \sum_{b=1}^B \left[ \Delta e_b \widetilde{m}_b + e_{b,\mathrm{train}} \Delta m_b + \Delta e_b \Delta m_b \right] - (\text{mean terms}).
\end{align*}
Using the Cauchy-Schwarz inequality, we bound the cross-terms. For example:
\[
\left| \frac{1}{B} \sum_{b=1}^B e_{b,\mathrm{train}} \Delta m_b \right| \le \sqrt{\frac{1}{B} \sum_{b=1}^B e_{b,\mathrm{train}}^2} \sqrt{\frac{1}{B} \sum_{b=1}^B (\Delta m_b)^2}.
\]
The first factor is bounded (by the LLN applied to $e_{b,\mathrm{train}}$), and the second factor converges to zero as shown above. Similar bounds apply to terms involving $\Delta e_b$ (which converges uniformly) and the product $\Delta e_b \Delta m_b$. The mean terms vanish similarly. Therefore, $|\widehat{C}_N - \widetilde{C}_N| \xrightarrow{p} 0$.

Finally, since $(e_{b,\mathrm{train}}, \widetilde{m}_b)$ are i.i.d. vectors across blocks $b$, the standard LLN implies $\widetilde{C}_N \xrightarrow{p} \mathrm{Cov}(e_{b,\mathrm{train}}, \widetilde{m}_b) = C_N$. Thus, $\widehat{C}_N \xrightarrow{p} C_N$.
\end{proof}

\begin{proof}[\textbf{Proof of Theorem~\ref{thm:CLT_varest}}]
By Lemmas~\ref{lem:train-var-est}, \ref{lem:test-var-est}, and~\ref{lem:cross-est},
\[
N \widehat{V}_{N,\mathrm{train}} \xrightarrow{p} N V_{N,\mathrm{train}}, \qquad \widehat{V}_{N,\mathrm{test}} \xrightarrow{p} V_{N,\mathrm{test}}, \qquad N \widehat{C}_N \xrightarrow{p} N C_N.
\]
From the variance decomposition \eqref{eq:asym-var}, we obtain
\[
\widehat{\sigma}_N^2 = N \widehat{V}_{N,\mathrm{train}} + \widehat{V}_{N,\mathrm{test}} + 2N \widehat{C}_N \xrightarrow{p} N V_{N,\mathrm{train}} + V_{N,\mathrm{test}} + 2N C_N = \sigma_N^2,
\]
establishing part~(1).

For part~(2), Theorem~\ref{thm:CLT-bCV} gives
$\sqrt{n}(\widehat{e}^{\mathrm{CV}} - e^a_N) \xrightarrow{d} \mathcal{N}(0, \sigma_N^2)$. 
Since $\widehat{\sigma}_N^2 \xrightarrow{p} \sigma_N^2 > 0$ by part (1), by the Slutsky's Theorem, we have
\[
\frac{\sqrt{n}(\widehat{e}^{\mathrm{CV}} - e^a_N)}{\widehat{\sigma}_N} \xrightarrow{d} \mathcal{N}(0, 1).
\]
\end{proof}

\section{Block-Out CV under Fixed-$B$ Asymptotics}
\label{sec:regimeB}

The asymptotic results in Section~\ref{subsec:bCV} rely on the regime where the training size $N$ is fixed while the number of blocks $B = n/N \to \infty$ (the fixed-$N$ asymptotic regime). While valid for time-series rolling windows or massive datasets, in many econometric applications, the number of fold splits $B$ is small and fixed (e.g., $B=5$ or $10$), while the sample size within each block $N$ is large.

We define \emph{fixed-$B$ asymptotic regime} as the asymptotic framework where $B$ is fixed and $N \to \infty$ (implying $n \to \infty$). Under this regime, the estimator $\widehat{e}^{\mathrm{CV}}$ aggregates losses from $B$ highly stable predictors. The variance estimator $\widehat{\sigma}^2_{N_k}$ based on the fixed-$N$ asymptotic regime can perform poorly here because it relies on estimating the between-block covariance structure from very few ($B$) observations.

In this section, we derive a Central Limit Theorem (CLT) for the fixed-$B$ asymptotic regime by exploiting the algorithmic stability of the learner as $N \to \infty$. This approach leverages recent results on cross-validation stability \citep{lei2025modern}, specifically the behavior of risk gradients for deterministic centering.

Recall that the block-out cross-validation estimator is given by:
\[
\widehat{e}^{\mathrm{CV}} = \frac{1}{n} \sum_{i=1}^{n} \widetilde{\ell}_i, \quad \text{where} \quad \widetilde{\ell}_i = \frac{1}{B-1} \sum_{b \neq b(i)} \ell\bigl(f_{N}^{(b)}(X_i), Y_i\bigr),
\]
where $b(i)$ denotes the block containing observation $i$. Note that $\widetilde{\ell}_i$ is the average loss of observation $i$ evaluated by the $B-1$ models that did not include $i$ in their training set. Let $e^a_N = \mathbb{E}[\ell(f_{D_N}(X), Y)]$ be the population prediction error for a training set of size $N$.

To establish asymptotic normality centered at $e^a_N$, we require conditions on the stability of the learning algorithm. Following \citet{lei2025modern}, we distinguish between the stability of the loss function and the stability of the risk (expected loss). Define the conditional risk of a model trained on dataset $D_N$ as $R(D_N) = \mathbb{E}_{Z}[\ell(f_{D_N}(X), Y)]$.

\begin{assumption}[Algorithm Stability]
\label{assu:RiskStable}
The algorithm satisfies:
\begin{equation}\label{eq:Stability}
 \|\nabla_1 R\|_{L_2}:=\sqrt{\mathbb{E}\left[\left(\nabla_1 R(D_N)\right)^2\right]} =o(N^{-1})  \text{ as } N \to \infty,   
\end{equation}
where $\nabla_1 R(D_N) = R(D_N) - R(D_N^{(1)})$ denotes the Perturb-One Stability, i.e., the change in conditional risk when one training point is replaced by an independent copy.\footnote{See Definition 2.3 in \cite{lei2025modern} for details.} 
\end{assumption}

\paragraph{Discussion on Stability.} Assumption~\ref{assu:RiskStable} corresponds to the condition in Proposition 4.11.(1) of \citet{lei2025modern}, where the variability induced by training data vanishes relative to the evaluation noise. \citet{lei2025modern} also provides arguments\footnote{See the arguments after Theorem 4.7 (pp.54-55) in \citet{lei2025modern}.} that in machine learning contexts, the risk function usually enjoys much better stability than the loss function, making risk stability requirement in our Assumption  \ref{assu:RiskStable} plausible.

\begin{theorem}[CLT under Fixed-$B$ Asymptotics]
\label{thm:CLT-regimeB}
Under Assumption~\ref{assu:RiskStable}, as $N \to \infty$ with $B$ fixed:
\[
\sqrt{n}\bigl(\widehat{e}^{\mathrm{CV}} - e^a_N\bigr) \xrightarrow{d} \mathcal{N}(0, \tau^2).
\]
\end{theorem}

\begin{proof}[\textbf{Proof of Theorem \ref{thm:CLT-regimeB}}] 
We adapt Proposition 4.11.(1) in \citet{lei2025modern}, along with the preceding analysis and lemmas \citet{lei2025modern}, to our block-out CV design.  We employ the H\'{a}jek projection technique for deterministic centering as formalized in \citet{lei2025modern}. Let $Z_i = (X_i, Y_i)$ and define the statistic $L_n = \sum_{i=1}^n \widetilde{\ell}_i$. The first-order H\'{a}jek projection is:
\[
\widehat{L}_n = \sum_{j=1}^n \mathbb{E}[L_n \mid Z_j] - (n-1)\mathbb{E}[L_n].
\]
We decompose the conditional expectation $\mathbb{E}[L_n \mid Z_j]$ into two components:
\begin{enumerate}
    \item Test component $j=i$: $Z_j$ acts as the test point for $\widetilde{\ell}_j$. Since $\widetilde{\ell}_j$ uses models trained on blocks not containing $j$, $Z_j$ is independent of the training sets. Thus:
    \[
    \mathbb{E}[\widetilde{\ell}_j \mid Z_j] = l_N(Z_j) \equiv \mathbb{E}[\ell(f_{D_N}(X), Y) \mid Z=Z_j].
    \]
    \item Training component $j \neq i$: $Z_j$ acts as a training point for observation $i$, which implies $j$ is in a training block used for $\widetilde{\ell}_i$). Hence, it affects the model prediction. In our block design, $Z_j$ is in the training set for all $B-1$ blocks that do not contain $j$. The sum of these influences corresponds to the function $g_N(Z_j)$ defined in \citet[Eq. 40]{lei2025modern}.
\end{enumerate}
The variance of the projection is $\mathrm{Var}(\widehat{L}_n) \approx n \mathrm{Var}(l_N(Z_1) + g_N(Z_1))$, where $g_N$ captures the sum of training influences. By the Efron-Stein inequality, the variance of the training influence scales with the stability of the risk:
\[
\mathrm{Var}(g_N(Z_1)) \le C \left( N \|\nabla_1 R(D_N)\|_{L_2} \right)^2.
\]
By Assumption~\ref{assu:RiskStable} (Risk Stability), we have $N \|\nabla_1 R\|_{L_2} \to 0$. Consequently, the total asymptotic variance is dominated by the test term:
\[
\frac{1}{n} \mathrm{Var}(\widehat{L}_n) \to \mathrm{Var}(l_N(Z_1)) = \tau^2.
\]
Standard results on U-statistics imply the remainder $L_n - \widehat{L}_n$ is $o_p(\sqrt{n})$. Applying the standard CLT to the sum of i.i.d. variables $l_N(Z_i)$ yields the result.
\end{proof}

In the fixed-$B$ asymptotic regime, since the training variance component vanishes, the asymptotic variance $\tau^2$ is simply the variance of the prediction error on a fresh test point. This can be estimated by the sample variance of the aggregated losses:
\begin{equation}
\label{eq:tau-hat-B}
\widehat{\tau}^2_B = \frac{1}{n-1} \sum_{i=1}^n \left( \widetilde{\ell}_i - \widehat{e}^{\mathrm{CV}} \right)^2.
\end{equation}

\begin{corollary}[Studentized CLT]\label{corol:CLT-regimeB}
Under the conditions of Theorem~\ref{thm:CLT-regimeB}, $\widehat{\tau}^2_B \xrightarrow{p} \tau^2$, and
\[
\frac{\sqrt{n}(\widehat{e}^{\mathrm{CV}} - e^a_N)}{\widehat{\tau}_B} \xrightarrow{d} \mathcal{N}(0, 1).
\]
\end{corollary}

\section{Details of ML Implementation}\label{sec:ML-detail}

\paragraph{Outcomes and feature construction.}
We apply the following pre-processing to all ML methods we use. For earnings as the target outcome, we winsorize it prior to analysis at the 1st and 99th percentiles to reduce
sensitivity to extreme outliers.\footnote{We alternatively train the ML models on $\log(1+Y)$ while calculating the error on the original level scale.} The continuous predictors are standardized, and the categorical predictors are handled as follows:
(i) rare categories are collapsed into a single ``other'' level when their
frequency falls below a minimum-count threshold (e.g., 50), and
(ii) the resulting categorical variables are converted to indicators
(with one category dropped as a reference group). In addition, we augment $X$ with other
outcomes as predictors (continuous outcomes enter as standardized numeric
features; categorical outcomes enter via indicators), excluding the
current target and respecting the earnings/wage exclusion noted above.

\paragraph{Learning-curve design and block-out CV.}
Fix a training size $N$. After randomly permuting the analysis sample once for
reproducibility, we partition the data into $B=\lfloor n/N \rfloor$ disjoint
blocks of size $N$. For each block $b=1,\ldots,B$, we fit the ML model using
only the observations in block $b$ as the training set, generate predictions
for all observations, and evaluate loss on the complement of block $b$. Averaging the out-of-block losses
across $b$ yields the block-out CV estimate of predictive risk at training
size $N$. Repeating this procedure over a grid of $N$ values produces the
error curve.

\paragraph{Model classes and per-$N$ hyperparameter tuning.}
For each $N$, we tune hyperparameters on a tuning subset of size
$N$ drawn at random from the analysis sample, and then apply the
selected hyperparameters in the block-out CV loop.\footnote{This approach is useful in stabilizing the error curve across $N$, although to faithfully follow the definition of $e^a_N$, we may want to tune hyperparameters for each training block in the block-out CV loop. Still, the current approach may be acceptable as the tuning is low-dimensional and held fixed across CV folds for a given N.} For regression (continuous $Y$), we consider lasso and random forest. The lasso penalty parameter is selected by \texttt{LassoCV} using 5-fold CV. The random forest uses 300 trees and is tuned by 3-fold CV over a small grid of structural parameters (maximum depth in $\{\texttt{None},10,20\}$ and minimum leaf size in $\{1,5\}$), with performance scored in terms of RMSE. For classification (discrete $Y$), we consider $\ell_1$-penalized logistic regression (logit lasso) and random forest classification. Logistic regression is fit using the \texttt{saga} solver and tuned over a grid of inverse-regularization parameters $C\in\{0.01,0.1,1,10\}$ using stratified CV; random forest classification uses 300 trees and is tuned over the same depth/leaf grid as in regression using stratified CV.

For neural networks used in the preliminary analysis in Section \ref{sec:PSID_more}, we fit a shallow fully-connected feedforward network with two ReLU hidden layers (64 and 32 units) and an output layer that is linear for regression and softmax for classification (with the number of outputs set to the number of classes observed in the training data). Networks are trained with the Adam optimizer using MSE for regression (100 epochs) and categorical cross-entropy for classification (10 epochs, with one-hot encoded labels).

\paragraph{Small-$N$ safeguards for classification.}
For very small $N$, stratified CV and/or block-level training sets can become
degenerate due to severe class imbalance. We therefore implement
explicit safeguards: (i) the number of folds in stratified CV is chosen
adaptively and never exceeds the minimum class count (capped at 3 folds), and
(ii) if a block-level training set contains only one class, we default to a
majority-class prediction rule for that block.

\paragraph{Uncertainty quantification.}
We report uncertainty measures derived from the variance estimator introduced in Section \ref{sec:InferenceTheory}: the
cross-block variability of observation-level out-of-block losses yields an
estimated variance component $\widehat{\sigma}_k^2$, which is converted into
standard errors of the CV risk via $\sqrt{\widehat{\sigma}_k^2/n}$. For $\widehat{\sigma}_k^2$, as mentioned in the main text, we use a hybrid approach based on the fixed-$N$ asymptotic theory (Theorem \ref{thm:CLT-bCV}) when $N\le400$ and the fixed-$B$ asymptotic theory (Corollary \ref{corol:CLT-regimeB} in Section \ref{sec:regimeB}) otherwise. For RMSE, we
apply a delta-method transformation from MSE to RMSE to obtain confidence
intervals.

\section{Additional Results for the Applications}\label{sec:PSID_more}

\subsection{Other Prediction Methods} \label{app:800}
Here we compare LLM performance to ML performance at a fixed training size for the ML models: 800 observations. In addition to ChatGPT-4, we report performance for  Meta's Llama~3 models (Llama3-8B cutoff: March 2023; Llama3-70B cutoff:
December 2023). We find that predictive performance is broadly similar across LLMs and, separately, across machine-learning algorithms, while varying systematically across outcomes. This pattern suggests that our main findings are not driven by the choice of a particular LLM or comparator algorithm.

\begin{figure}
\noindent \begin{centering}
\includegraphics[scale=0.3]{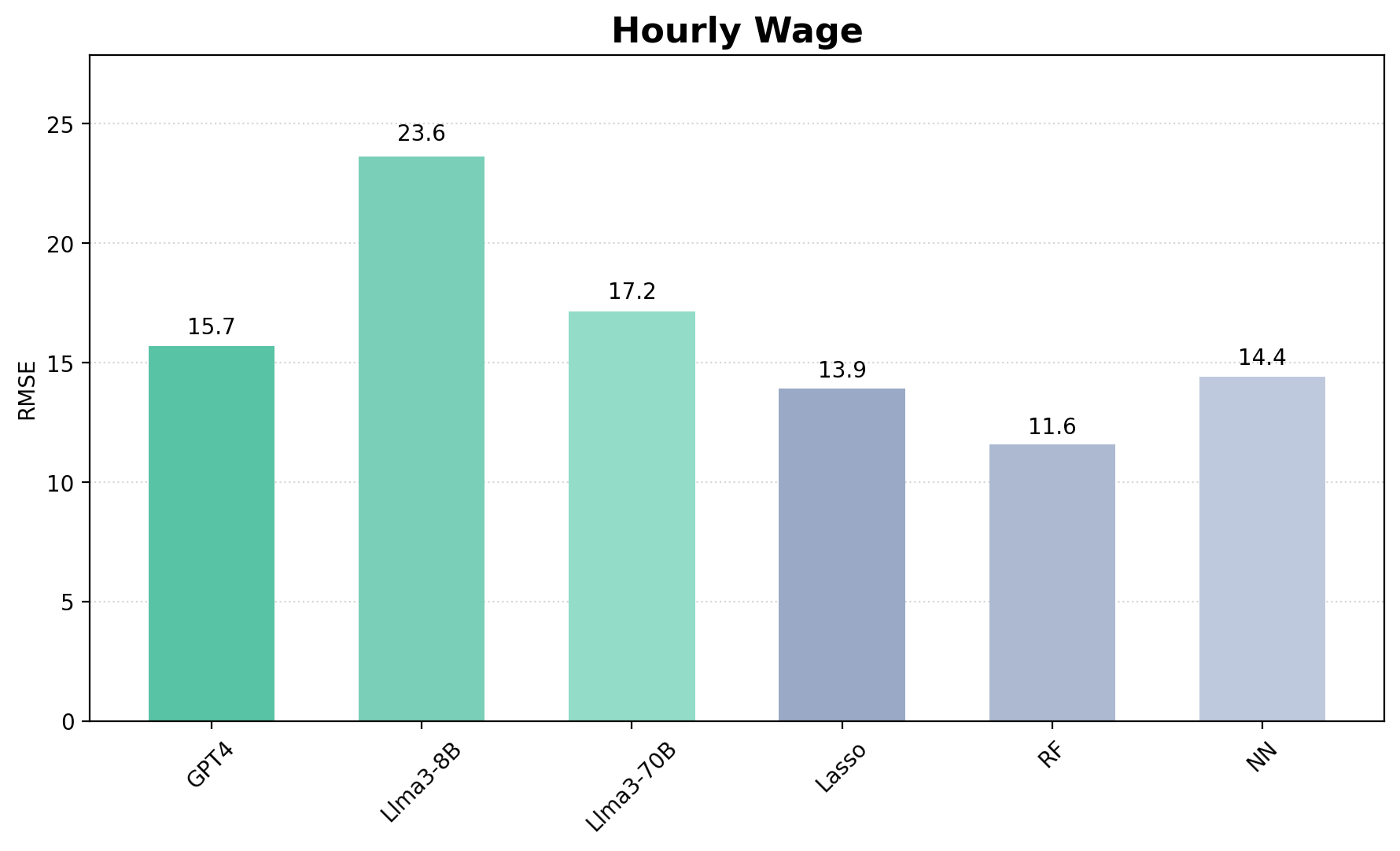}
\par\end{centering}
\caption{Root Mean Squared Error (lower is better) ($N=800$)}
\label{fig:regression}
\end{figure}

\begin{figure}
\noindent \begin{centering}
\includegraphics[scale=0.35]{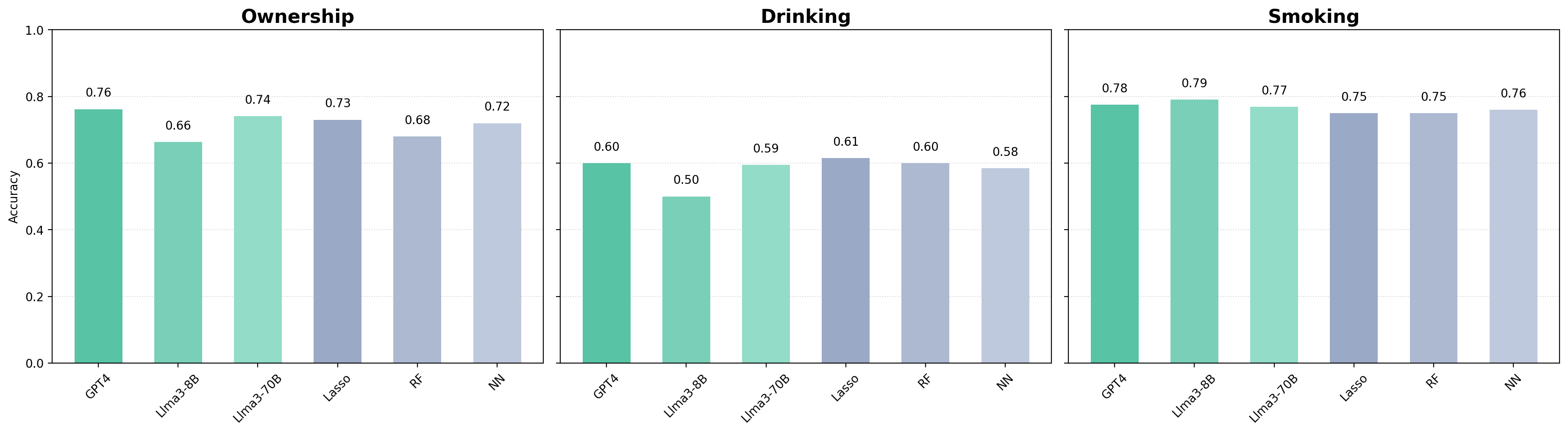}
\par\end{centering}
\caption{Accuracy (higher is better) ($N=800$)}
\label{fig:classification}
\end{figure}

\subsection{Other Error Metrics}

In the main text, prediction error for discrete outcomes is measured using the 0-1 loss function. Misclassification rate can be mechanically low for outcomes that are very skewed, which might mask differences between the LLMs and ML algorithms that would otherwise emerge. As a robustness check, we thus compare performance in terms of \emph{balanced accuracy} and \emph{F1}, defined as follows.

For discrete (binary) outcomes, let $Y \in \{0,1\}$ denote the true label and $f(X) \in \{0,1\}$ the predicted label. Writing $\mathrm{TP}$, $\mathrm{TN}$, $\mathrm{FP}$, and $\mathrm{FN}$ for the usual entries of the confusion matrix, the \emph{balanced accuracy} is defined as the average of the class-wise recall,
\[
\mathrm{BA}(f)
= \tfrac{1}{2}\!\left(
\frac{\mathrm{TP}}{\mathrm{TP}+\mathrm{FN}}
+
\frac{\mathrm{TN}}{\mathrm{TN}+\mathrm{FP}}
\right).
\]
The \emph{F1 score} is defined as the harmonic mean of precision and recall,
\[
\mathrm{F1}(f)
= \frac{2\,\mathrm{TP}}{2\,\mathrm{TP}+\mathrm{FP}+\mathrm{FN}}
= \frac{2\,\mathrm{Precision}\times \mathrm{Recall}}{\mathrm{Precision}+\mathrm{Recall}},
\]
and emphasizes performance on the positive class by jointly penalizing false positives and false negatives. Figures \ref{fig:classification_balanced} and \ref{fig:classification_f1} replicate our classification results in Appendix \ref{app:800} using these error metrics. We find that the performances across domains are qualitatively similar to those in Figure \ref{fig:classification}.

\begin{figure}
\noindent \begin{centering}
\includegraphics[scale=0.35]{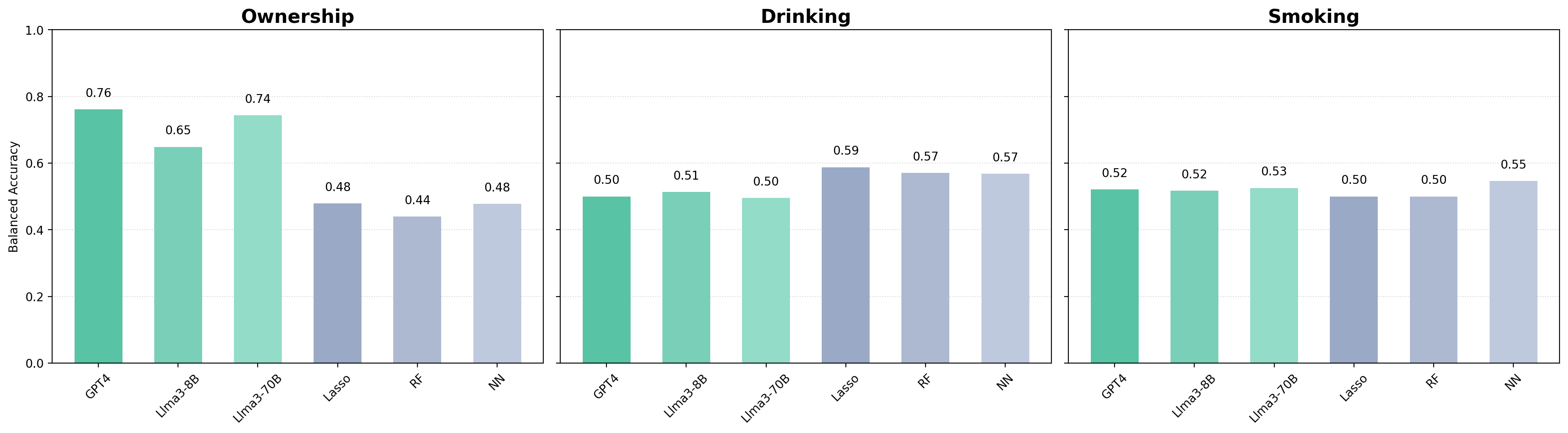}
\par\end{centering}
\caption{Balanced Accuracy (higher is better) ($N=800$)}
\label{fig:classification_balanced}
\end{figure}

\begin{figure}
\noindent \begin{centering}
\includegraphics[scale=0.35]{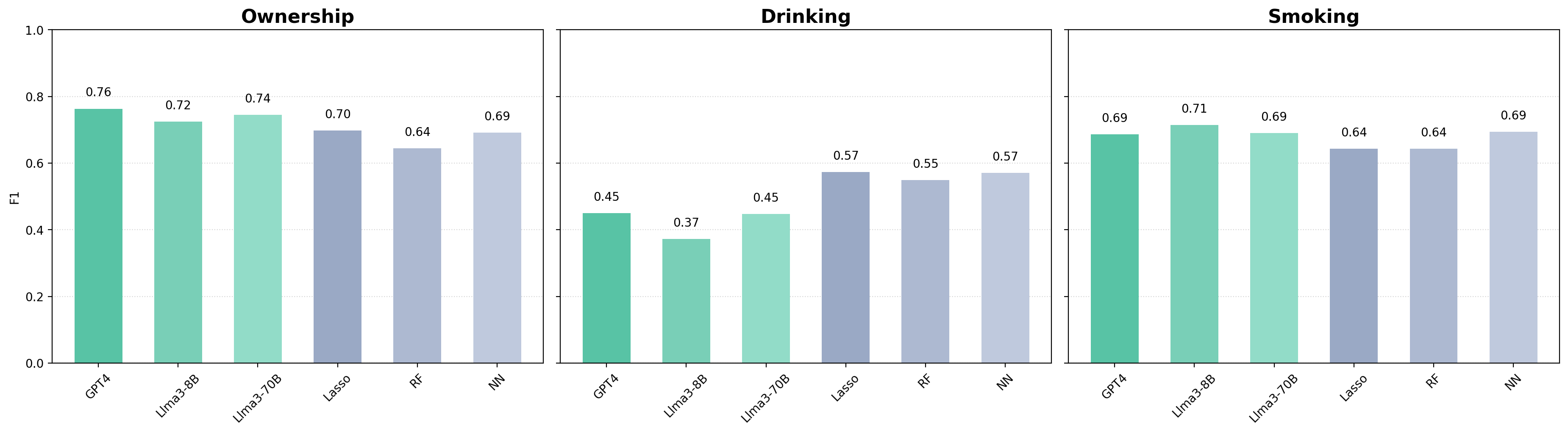}
\par\end{centering}
\caption{F1 (higher is better) ($N=800$)}
\label{fig:classification_f1}
\end{figure}

\end{appendix}

\end{document}